\documentclass[12pt,reqno]{amsart}
\usepackage{graphicx}
\usepackage{amssymb,amsmath}
\usepackage{amsthm}
\usepackage{color}
\usepackage[pdf]{pstricks}
\usepackage{hyperref}
\usepackage{empheq}
\usepackage{pgfplots}
\usepackage{comment}
\usepackage{booktabs} 
\usepackage{multirow}
\usepackage{array}
\usepackage{cancel}

\newtheorem{remark}{Remark}
\newtheorem{example}{Example}
\newtheorem{lemma}{Lemma}

\newtheorem{prop}{Proposition}

\newtheorem{theorem}{Theorem}

\newcommand{\ga}[1]{\textcolor{blue}{{#1}}}

\begin{document}
		\title[Stability of ILMs in the lattice with competing powers]{Stability of intrinsic localized modes on the 
			lattice with competing power nonlinearities}
	
	\author[G. L. Alfimov]{Georgy L. Alfimov}
	\address[G. L. Alfimov]{National Research University of Electronic Technology MIET, Zelenograd, Moscow 124498, Russia}
	\email{galfimov@yahoo.com}
		
	\author[P. A. Korchagin]{Pavel A. Korchagin}
	\address[P. A. Korchagin]{LLC “LSRM”, Zelenograd, Moscow 124482, Russia}
	\email{korchagin@lsrm.ru}
	
	\author[D. E. Pelinovsky]{Dmitry E. Pelinovsky}
	\address[D. E. Pelinovsky]{Department of Mathematics and Statistics, McMaster University, Hamilton, Ontario, Canada, L8S 4K1}
	\email{pelinod@mcmaster.ca}
	
	\begin{abstract}
		We study the discrete nonlinear Schr\"{o}dinger equation with competing powers $(p,q)$ satisfying $2 \leq p < q$. The physically relevant cases are given by $(p,q) = (2,3)$, $(p,q) = (3,4)$, and $(p,q) = (3,5)$. In the anticontinuum limit, all intrinsic localized modes are compact and can be classified by their codes, which record one of two nonzero (smaller and larger) states and their sign alternations. By using the spectral stability analysis, we prove that the codes for larger states of the same sign are spectrally and nonlinearly (orbitally) stable, whereas the codes for smaller states of the alternating signs are spectrally stable but have eigenvalues of negative Krein signature. We also identify numerically the spectrally stable codes which consist of stacked combinations of the sign-definite larger states and the sign-alternating smaller states. 
	\end{abstract}   
	
	\maketitle

% Possible journals: 1) SIAD, 2) Nonlinearity, 3) Physica D, 4) DCDS, 5) Chaos

\section{Introduction}\label{Sec:Problem}

The Discrete Nonlinear Schr\"odinger (DNLS) equation, 
\begin{align}
    i\frac{d\Psi_n}{d t}+C\Delta_2\Psi_n+\left|\Psi_n\right|^2\Psi_n=0, \label{Eq:DNLS}
   & .
\end{align}
where $\Psi_n(t) \in \mathbb{C}$, $t \in \mathbb{R}$, $n \in \mathbb{Z}$, $C > 0$, and $\Delta_2\Psi_n \equiv \Psi_{n+1}-2\Psi_n+\Psi_{n-1}$, 
is a fundamental lattice model that has emerged in various physical contexts (see, e.g., the  surveys \cite{HTs99,EJ2003} and the book \cite{K2009}). In nonlinear optics, this equation  describes arrays of parallel Kerr waveguides fabricated on a common substrate \cite{ChLS03,DiscrSol08}.  In  Bose-Einstein condensate (BEC) theory, it is used to describe the dynamics of the condensate cloud in a sigar-shape optical trap \cite{TS01,KB04,MO06}. In this context, the lattice equation (\ref{Eq:DNLS}) has been derived from the Gross-Pitaevskii equation with periodic potential by using the basis of Wannier functions \cite{AKKS04}, see also the book \cite{Pel11}.  

A particular attention is given to specific solutions of the DNLS equation which are spatially localized and periodic in time. These solutions are called {\it intrinsic localized modes} (ILMs) or, alternatively, {\it  discrete solitons}. The simplest class of  ILMs is represented by the stationary oscillations of the frequency $\omega$ written in the form
\begin{gather}
\Psi_n(t) = \Phi_n e^{i\omega t},\quad \lim_{n \to \pm \infty} \Phi_n = 0. \label{Eq:SolitonCond}
\end{gather}
The existence and stability of ILMs have been extensively studied in the DNLS equation in the anticontinuum limit (ACL) (see, e.g., \cite{ABK04,PKF05,Yo16}). The ACL was introduced in the context of 
breathers of the discrete Klein--Gordon equations in \cite{MA1994,A1997} 
but it was found particularly powerful in the context of the DNLS equation. 

Many other physically relevant lattice models differ from the cubic 
DNLS equation (\ref{Eq:DNLS}) by their type of nonlinearity. Dynamics of electromagnetic field in array of waveguides in photorefractive crystals is described by the DNLS equation with a {\it saturable} nonlinearity \cite{HMSK04}.  Approximation of the saturable nonlinearity by two power terms yields the cubic--quintic DNLS equation:
\begin{gather}
    i\frac{d\Psi_n}{d t}+C\Delta_2\Psi_n+\varkappa\left|\Psi_n\right|^2\Psi_n-\Gamma\left|\Psi_n\right|^4\Psi_n=0. 
    \label{Eq:DNLS-CQ}
\end{gather} 
If  $\varkappa$ and $\Gamma$ are positive constants, Eq.~(\ref{Eq:DNLS-CQ}) includes two competing powers of opposite signs (referred to as DNLS with {\it competing nonlinearity}).  It is known for $C > 0$ that the ``plus'' term $\sim\left|\Psi_n\right|^2\Psi_n$ is a {\em focusing} nonlinearity and the ``minus'' term $\sim\left|\Psi_n\right|^4\Psi_n$ is a {\em defocusing} nonlinearity. The interplay of these two factors results in new features of the model with either saturable or cubic--quintic nonnlinearity, e.g.  multistability of steady-states \cite{CTCM2006,TD10,CP11} and mobility of localized excitations \cite{MHM07,MHM08,MVM13,AT2019,MCKC06,MCKC08,AKLP19}.

Recently, two more versions of DNLS with competing nonlinearity have received wide discussion in the physical literature. Both of them have been used to describe a mixture of two BECs in the presence of quantum fluctuations (QFs). It was shown in \cite{P15} that the correction to  the mean-field energy due to QFs can stabilize the BEC mixture.  The correction depends on the dimensionality of the system \cite{AP16}.  Theoretical predictions were confirmed by experiments,  in 2D \cite{Experiment01} and 3D \cite{Experiment02} cases.

(i) In the case of 1D systems (i.e. assuming that the trap is  very narrow in the transverse direction), the correction due to QFs gives an additional {\it attractive} nonlinear term $\sim |\Psi|\Psi$ to the cubic repulsive term \cite{ZYZZ2021,KZ2024}. Expanding the wavefunction with respect to the basis of Wannier functions \cite{AKKS04}, one arrives at the lattice equation \cite{Droplets01,Droplets02,SuzantoStudies}
\begin{gather}
    i\frac{d\Psi_n}{d t}+C\Delta_2\Psi_n+\varkappa\left|\Psi_n\right|\Psi_n-\Gamma\left|\Psi_n\right|^2\Psi_n=0,
    \label{Eq:Droplets01}
\end{gather}
where  $C$, $\varkappa$ and $\Gamma$ are positive. 

(ii) If the BEC is loaded in the 3D cigar-type (elongated) trap, then the correction due to QFs gives an additional {\it repulsive} nonlinear term $\sim |\Psi|^3\Psi$ to the cubic attractive term \cite{Edm20,Droplets03,Deb22,Nat22,AKA25}.  Averaging the 3D wavefunction over the transverse dimensions and applying the Wannier functions expansion (see \cite{Droplets03} for detail) leads to the lattice equation
\begin{gather}
    i\frac{d\Psi_n}{d t}+C\Delta_2\Psi_n+\varkappa\left|\Psi_n\right|^2\Psi_n-\Gamma\left|\Psi_n\right|^3\Psi_n=0,\label{Eq:Droplets02}
\end{gather}
where  $C$, $\varkappa$ and $\Gamma$ are positive. 

The DNLS equations (\ref{Eq:DNLS-CQ}), (\ref{Eq:Droplets01}), and (\ref{Eq:Droplets02}) are the particular cases of a more general model:
\begin{gather}
    i\frac{d\Psi_n}{d t}+C\Delta_2\Psi_n+ \varkappa\left|\Psi_n\right|^{p-1}\Psi_n-\Gamma\left|\Psi_n\right|^{q-1}\Psi_n=0,\label{Eq:Gen_pq}
\end{gather}
where $C$, $\varkappa$, and $\Gamma$ are positive constants, whereas $p,q\in \mathbb{N}$ satisfy $2 \leq p < q$. Since the wave function must remain spatially localized, ILMs are natural objects of the study. Assuming the solution of the form (\ref{Eq:SolitonCond}) with the $t$-dependent amplitudes $\{\Phi_n(t)\}_{n \in \mathbb{Z}}$, it follows from (\ref{Eq:Gen_pq}) that
\begin{gather}
    i\frac{d\Phi_n}{d t} + C\Delta_2 \Phi_n -\omega \Phi_n + \varkappa\left|\Phi_n\right|^{p-1} \Phi_n - \Gamma \left|\Phi_n\right|^{q-1} \Phi_n=0.
    \label{Eq:Gen_psi}
\end{gather}
If $\omega > 0$, we apply the time rescaling $t\to \omega t$ in  Eq.~\eqref{Eq:Gen_psi}, followed by the transformation
\begin{gather*}
    \varepsilon = \frac{C}{\omega}, \quad \gamma = \frac{\Gamma}{\omega} \left(\frac{\omega}{\varkappa}\right)^{\frac{q-1}{p-1}}, \quad \Phi_n = \left(\frac{\omega}{\varkappa}\right)^{\frac{1}{p-1}} u_n,
\end{gather*}
to obtain the normalized DNLS equation:
\begin{gather}
    i\frac{du_n}{d t} + \varepsilon\Delta_2u_n- u_n+\left|u_n\right|^{p-1} u_n-\gamma\left|u_n\right|^{q-1} u_n=0.
    \label{Eq:u_gen_pq_t}
\end{gather}
Steady states for the time-dependent DNLS equation~(\ref{Eq:u_gen_pq_t}) satisfy the difference equation
\begin{gather}
    \varepsilon\left(u_{n+1}-2u_n+u_{n-1}\right)- u_n+\left|u_n\right|^{p-1} u_n-\gamma\left|u_n\right|^{q-1} u_n=0. 
    \label{Eq:u_gen_pq}
\end{gather}
We represent the sequence $\{ u_n \}_{n \in \mathbb{Z}}$ for solutions to  Eq.~(\ref{Eq:u_gen_pq}) as a bi-infinite vector ${\bf u}=(\ldots u_{-1},u_0,u_1,\ldots)$. It follows from (\ref{Eq:u_gen_pq}) that 
\begin{gather*}
J=\overline{u}_nu_{n+1}-u_n\overline{u}_{n+1}
\end{gather*}
does not depend on $n$, where the bar means the complex conjugation. 
If $\lim\limits_{n \to \pm \infty} u_n = 0$, then $J=0$ so that either $u_n=0$ or
\begin{gather*}
\frac{u_{n+1}}{u_{n}}=\frac{\overline{u}_{n+1}}{\overline{u}_{n}},
\end{gather*}
and the arguments of $u_{n+1}$ and $u_n$ are equal modulo $\pi$.  Therefore, without loss of generality  we can assume that $u_n \in\mathbb{R}$ for any $n\in\mathbb{Z}$ so that ${\bf u} \in \mathbb{R}^{\mathbb{Z}}$.

In this study, we address the problem of stability of ILMs in the normalized DNLS equation with the competing nonlinearity (\ref{Eq:u_gen_pq_t}). The stability property of the ILMs is very important for physical applications, since only stable objects can be observed in experiments.  In a general formulation, this problem can hardly be solved analytically. For the  DNLS equation (\ref{Eq:DNLS}), the stability of solutions under arbitrary value of the coupling parameter $C$ can only be established numerically \cite{K2009}. However, the stability problem can be analyzed in the ACL for small values of the coupling parameter \cite{PKF05}. The method of  \cite{PKF05} was extended in \cite{RNKF09} to wider class of DNLS-type equations including the saturable nonlinearity. It is the purpose of this work to analyze the stability problem for the competing nonlinearity. The particular cases of the model (\ref{Eq:Gen_psi}), namely the cases $(p,q) = (2,3)$ in (\ref{Eq:Droplets01}) 
and $(p,q) = (3,4)$ in (\ref{Eq:Droplets02}), were recently studied in \cite{AKA25}, where bifurcations of ILMs were classified numerically, 
and in \cite{SuzantoStudies}, where stability of dark solitons was considered 
in the ACL analytically and numerically. 

Stability of ILMs is related to the question of minimization of the energy $H({\bf u})$ with or without a constraint of fixed mass $Q({\bf u})$, where
\begin{align}
\label{energy}
H({\bf u}) &= \sum_{n \in \mathbb{Z}} \varepsilon |u_{n+1}-u_n|^2  - \frac{2}{p+1} |u_n|^{p+1} + \frac{2 \gamma}{q+1} |u_n|^{q+1}, \\
\label{mass}
Q({\bf u}) &= \sum_{n \in \mathbb{Z}} |u_n|^2. 
\end{align} 
Both the energy and mass are conserved quantities of the DNLS equation (\ref{Eq:u_gen_pq_t}). The difference equation (\ref{Eq:u_gen_pq}) is the Euler--Lagrange equation for the critical points of the augmented energy 
\begin{equation}
\label{aug-energy}
\Lambda({\bf u}) = H({\bf u}) + Q({\bf u}).
\end{equation}
We distinguish between the spectral stability of ILMs, for which the spectrum of a linearized operator is a subset of $i \mathbb{R}$, and the nonlinear (orbital) stability of ILMs, for which perturbations to the orbit $\{ e^{i \alpha} {\bf u}\}_{\alpha \in \mathbb{R}}$ do not grow in time. See the book \cite{GP25} for the introduction to the stability analysis of ILMs.

The main analytical result of this study is the following theorem.

\begin{theorem}
	\label{th-main} 
	Let $p,q \in \mathbb{N}$ be fixed such that $2 \leq p < q$. For every $\gamma \in (0,\gamma_{{p,q}})$, where 
	 \begin{gather}
	\gamma_{{p,q}} \equiv \left(\frac{p-1}{q-1}\right)\cdot\left(\frac{q-p}{q-1}\right)^{\frac{q-p}{p-1}};\label{Eq:BoundGamma}
	\end{gather}
	there exists $\varepsilon_0 > 0$ and $C_0 > 0$ such that for every $\varepsilon \in (0,\varepsilon_0)$, there exists a solution ${\bf u} \in \ell^2(\mathbb{Z})$ of the difference equation (\ref{Eq:u_gen_pq}) such that 
	\begin{equation}
	\label{bound-ILMs}
	\| {\bf u} - {\bf u}^{(0)} \|_{\ell^2(\mathbb{Z})} \leq C_0 \varepsilon,
	\end{equation}	
where ${\bf u}^{(0)} = (\ldots,0,0,\tilde{\bf u},0,0,\ldots)$ with $\tilde{\bf u} \in \mathbb{R}^N$ of any length $N \geq 1$ defined by nonzero roots $\pm a$ and $\pm A$ with $0 < a  < A$ of the function 
\begin{equation}
\label{function}
\mathbb{R} \ni u \to f(u) = u (1 - |u|^{p-1} + \gamma |u|^{q-1}) \in \mathbb{R}.
\end{equation} 
Moreover, we have for every $\varepsilon \in (0,\varepsilon_0)$,
\begin{itemize}
	\item[(A)] The solution ${\bf u}$ with either
	$\tilde{\bf u} = (+A,+A,\dots, +A)$ or  $\tilde{\bf u} = (-A,-A,\dots,-A)$ is a minimizer of augmented energy $\Lambda$, hence it is spectrally and orbitally stable.

\item[(B)] The solution ${\bf u}$ with either 
$\tilde{\bf u} = (+a,-a,\dots,\pm a)$ or $\tilde{\bf u} = (-a,+a,\dots,\mp a)$ is spectrally stable but it is a constrained minimizer of energy $H$ for fixed mass $Q$ (which is orbitally stable) if and only if $N = 1$.  

\item[(C)] If 
\begin{equation}
\label{crit-1}
\frac{a^2}{f'(a)} + \frac{(N-1) A^2}{f'(A)} < 0,
\end{equation}
then the solution ${\bf u}$ with $\tilde{\bf u}$ consisting of either $(N-1)$ elements $(+A,+A,\dots,+A)$ and one element $+a$ or $(N-1)$ elements $(-A,-A,\dots,-A)$ and one element $-a$ (in any order) is a constrained minimizer of energy $H$ for fixed mass $Q$, hence it is spectrally and orbitally stable.

\item[(D)] If 
\begin{equation}
\label{crit-2}
\frac{(N-1) a^2}{f'(a)} + \frac{A^2}{f'(A)} > 0,
\end{equation}
then the solution ${\bf u}$ with $\tilde{\bf u}$ consisting of either $(N-1)$ elements  $(+a,-a,\dots,\pm a)$ and one element $\pm A$ or $(N-1)$ elements $(-a,+a,\dots,\mp a)$ and one element $\mp A$ (in any order but preserving the sign alternation) is spectrally stable but it is not a constrained minimizer of energy $H$ for fixed mass $Q$ for $N \geq 2$.
\end{itemize}
\end{theorem}

\begin{remark}
	The existence part of Theorem \ref{th-main} is proven as Proposition \ref{prop-existence}. The stability part of Theorem \ref{th-main} with items (A), (B), (C), and (D) is proven as Propositions \ref{prop-1}, \ref{prop-2}, \ref{prop-3}, and \ref{prop-4}. Each proposition also implies that all other ILMs of the same class as in (A), (B), (C), (D) are spectrally unstable, e.g. (A) the solution ${\bf u}$ with 
	$\tilde{\bf u} = (\pm A,\pm A,\dots, \pm A)$ are spectrally unstable if there exists at least one sign alternation, etc. 
\end{remark}

To simplify the formalism, we will use the code 
$\mathcal{A} = (A_+a_-A_-a_+)$ with $a_{\pm} = \pm a$ and $A_{\pm} = \pm A$ 
instead of the vector $\tilde{u} = (+A,-a,-A,+a)$. In this way, we can label 
the whole branch of ILMs for $\varepsilon \in (0,\varepsilon_0)$ by $\mathcal{A}$, whereas the 
vector $\tilde{\bf u}$ is only relevant for the limiting vector ${\bf u}^{(0)}$ as $\varepsilon \to 0$, see (\ref{bound-ILMs}).  

Items (A) and (B) of Theorem \ref{th-main} suggest the universal stability of the following codes of any length $N \geq 1$:
\begin{gather}
\mathcal{A}_A=(A_+A_+\ldots A_+),\quad \mathcal{A}_a=(a_+a_-\ldots a_\pm). \label{Eq:uniform}
\end{gather}
Items (C) and (D) include the first members of the sequence of stacked modes:
\begin{gather}
{\mathcal A}_{n,m}^+=(\underbrace{A_+A_+\ldots A_+}_{n} \underbrace{a_+a_-\ldots a_\pm}_{m}),\quad {\mathcal A}_{n,m}^-=(\underbrace{A_+A_+\ldots A_+}_{n} \underbrace{a_-a_+\ldots a_\pm}_{m}). \label{Eq:Stacked}
\end{gather}
%Since $f'(a) < 0$ and $f'(A) > 0$, the stability criteria (\ref{crit-1}) and (\ref{crit-2}) are generally satisfied for $N = 2,3$ for all values of $\gamma \in (0,\gamma_{p,q})$.
We have performed the full numerical study of spectrally stable codes, focusing on the physically relevant cases $(p,q) = (2,3)$, $(p,q) = (3,4)$, and $(p,q) = (3,5)$. Our numerical finding is that for values of $\gamma$ close to $\gamma_{p,q}$, there exist universally stable codes $\mathcal{A}^-_{k+1,k}$ for odd $N = 2k + 1$ in addition to universally stable codes $\mathcal{A}_A$ and $\mathcal{A}_a$ in (\ref{Eq:uniform}). Furthermore, for even $N=2k$ and for values of $\gamma$ close to $\gamma_{p,q}$ the ILM with codes $\mathcal{A}^+_{k,k}$ are stable for large values of $q$ as follows: for $q > 7$ in the case $p=2$, 
for $q > 5$ in the case  $p=3$, and for any $q > p$ in the case of $p \geq 4$.

In addition, we discovered many stable configurations for small values of $\gamma$. Eigenvalues of the spectral stability problem for small $\gamma > 0$ are very different in magnitude since $a \to 1$ and $A \to \infty$ as $\gamma \to 0$, which explains stability of many codes in addition to $\mathcal{A}_A$, $\mathcal{A}_a$ in (\ref{Eq:uniform}), and their stable stacked versions described above. We defer for further study the asymptotic analysis of the spectral stability of ILMs in the limits $\gamma \to \gamma_{p,q}$ and $\gamma \to 0$.

The rest of this paper is organized as follows. Existence of ILMs in the ACL is studied in Section \ref{Sect:StatStates}, where the first part of Theorem \ref{th-main} is proven. We divide all codes into groups of {\em equivalent} and {\em irreducible} codes and give the count of all irreducible codes for $N \geq 1$. 

Stability of ILMs in the anticontinuum limit is studied in Section \ref{Sec:StabilityACL}, where the second part of Theorem \ref{th-main} is proven by using analysis of the truncated spectral stability problem and persistence of eigenvalues in the full spectral stability problem. Moreover, we combine the classification of the spectrally stable codes with the count of negative eigenvalues in the second variation of the augmented energy $\Lambda$ and in the second variation of the constrained energy $H$ for fixed mass $Q$, which are 
included to get conclusions on the nonlinear (orbital) stability of ILMs.

Numerical results visualizing simplest ILMs and their bifurcations are shown in Section \ref{Sect:Continuation}. Numerical results about spectral stability 
of ILMs are summarized in Section \ref{Sect:Total}. We give the complete count of spectrally stable ILMs in tables and show bifurcations of stable and unstable eigenvalues for some ILMs in figures. We focus on the physically relevant cases $(p,q) = (2,3)$, $(p,q) = (3,4)$, and $(p,q) = (3,5)$, but also give results for $(p,q) = (3,6)$ for completeness of presentation. Section \ref{Sect:Conclusion} concludes the paper with summary and open problems for future study.

\section{Existence of ILMs in the anticontinuum limit}
\label{Sect:StatStates}

We explore the anticontinuum limit (ACL) in the difference equation (\ref{Eq:u_gen_pq}), which corresponds to the limit of uncoupled lattice sites as $\varepsilon \to 0$. If $\varepsilon = 0$, then $u_n$ for each $n \in \mathbb{Z}$ is a real root of $f(u) : \mathbb{R} \to \mathbb{R}$ given by (\ref{function}). We have 
\begin{itemize}
    \item[(a)] three real roots $\{0,+1,-1\}$ of $f$ if $\gamma=0$;
    \item[(b)] five real roots $\{0,+a,-a,+A,-A\}$ of $f$ with $0<a<A$ if $\gamma \in (0,\gamma_{p,q})$, where $\gamma_{p,q}$ is given by (\ref{Eq:BoundGamma}); 
    \item[(c)] three real roots $\{ 0, +u_{p,q}, -u_{p,q}\}$ of $f$ if $\gamma=\gamma_{{p,q}}$, where
    \begin{gather*}
        u_{p,q} \equiv \left(\frac{q-1}{q-p}\right)^{1/(p-1)}.
    \end{gather*}
    \item[(d)] one real root $\{ 0 \}$ of $f$ if $\gamma>\gamma_{{p,q}}$.
\end{itemize}

In case (d), no spatially decaying solutions of the difference equation (\ref{Eq:u_gen_pq}) exist for any $\varepsilon>0$. Indeed, if we multiply (\ref{Eq:u_gen_pq}) by $\bar{u}_n$ and sum over $n \in \mathbb{Z}$, we obtain 
$$
-\varepsilon \sum_{n \in \mathbb{Z}} |u_{n+1}-u_n|^2 - \sum_{n \in \mathbb{Z}} |u_n|^2 (1-|u_n|^{p-1} + \gamma |u_n|^{q-1}) = 0,
$$
where the left-hand side is strictly negative for nonzero solutions if $\gamma > \gamma_{p,q}$ and $\varepsilon > 0$. 

Cases (a) and (c) represent boundaries of the interval $(0,\gamma_{{p,q}})$ which are sensitive to perturbations. Therefore, in this study, we focus on case (b) and assume $\gamma \in (0,\gamma_{{p,q}})$. We obtain the following elementary result.

\begin{lemma}
	\label{lem-1}
	For any $p,q$ with $q>p$, and $0<\gamma<\gamma_{{p,q}}$, we have  $f'(a) < 0$ and $f'(A) > 0$.
\end{lemma}

\begin{proof}
	
	If $0<\gamma<\gamma_{{p,q}}$, then there are two positive roots of $f$ at $a$ and $A$. It is obvious that $f(u) > 0$ for $u \in (0,a) \cup (A,+\infty)$ and $f(u) < 0$ for $u \in (a,A)$. This implies that
	\begin{align*}
	&f'(a)=1-pa^{p-1}+\gamma qa^{q-1}\leq 0, \\
	&f'(A)=1-pA^{p-1}+\gamma qA^{q-1}\geq 0.
	\end{align*}
	To show that $f'(a) \neq 0$, we check that the system of equations
	\begin{gather*}
	f(a)=0,\quad f'(a)=0
	\end{gather*}
	has a solution on $(0,\infty)$ if and only if $\gamma=\gamma_{{p,q}}$ (in which case $a=A$). The same argument applies to show that $f'(A) \neq 0$. Hence, we have $f'(a) < 0$ and $f'(A) > 0$.  
\end{proof}

Let ${\bf u}^{(0)} \in \mathbb{R}^{\mathbb{Z}}$ be a bi-infinite solution of the difference equation (\ref{Eq:u_gen_pq}) when  $\varepsilon=0$.  For $\gamma \in (0,\gamma_{{p,q}})$, each element of ${\bf u}^{(0)}$ may take any of the values $\{0,+a,-a,+A,-A\}$ independently of others. Since we consider ILMs, we assume that  ${\bf u}^{(0)}$ contains a {\it finite number} of nonzero components $\pm a$ and $\pm A$. {\em Moreover, we assume that the ILM at the ACL contains only $N$ nonzero components in the $N$ consequent elements of ${\bf u}^{(0)}$}. Without loss of generality, we allocate the nonzero components between $n = 1$ and $n = N$ and introduce the vector $\tilde{\bf u} = (u_1^{(0)},u_2^{(0)},\ldots,u_N^{(0)}) \in \mathbb{R}^N$ such that 
${\bf u}^{(0)} = (\ldots,0,0,\tilde{\bf u},0,0,\ldots)$. For the vector $\tilde{\bf u}$, we also introduce the code $\mathcal{A}$ with symbols $a_{\pm} = \pm a$ and $A_{\pm} = \pm A$. 

\begin{remark} 
	If ${\bf u}$ is a solution of (\ref{Eq:u_gen_pq}) for any $\varepsilon \in \mathbb{R}$, then ${\bf R} \mathbf{u}$, $-{\mathbf{u}}$, and $-{\bf R} \mathbf{u}$ are also solutions of (\ref{Eq:u_gen_pq}) for any $\varepsilon \in \mathbb{R}$, 
	where ${\bf R}$ is the reversibility operator given by $({\bf R u})_n = u_{-n}$. This is due to reversibility of $\Delta_2$ about any node $n \in \mathbb{Z}$ and the sign symmetry of the nonlinear terms. Existence and stability of each of the four equivalent ILMs are identical to each other. We only pick one of the four equivalent ILMs and call it the irreducible ILM.
\end{remark}

To count the total number of irreducible ILMs of the same length $N$, 
we have the following elementary result.

\begin{lemma}
	\label{lem-2}
For each odd $N = 2k + 1$ with $k \in \mathbb{N}$, there exist 
$$ 
16^k+ 4^k
$$ 
irreducible codes of length $N$. For each even $N = 2k$ with $k \in \mathbb{N}$, there exist 
$$
\frac14 (16^k+2\cdot 4^k)
$$ 
irreducible codes of length $N$. 
\end{lemma}

\begin{proof}
Given that the zero symbol is not used in the codes, each symbol must be one of four types: $A_+$, $A_-$, $a_+$, or $a_-$, and there exists no vector $\tilde{\bf u}$ satisfying $\tilde{\bf u} = - \tilde{\bf u}$. We say that the vector $\tilde{\bf u}$ is {\it ${\bf R}$-symmetric} if ${\bf R \tilde{u}}= \tilde{\bf u}$. Also we say that $\tilde{\bf u}$ is {\it $-{\bf R}$-symmetric} if $-{\bf R \tilde{u}} = \tilde{\bf u}$. No vectors can be both ${\bf R}$-symmetric and $-{\bf R}$-symmetric.

Let $N = 2k+1$, $k \in \mathbb{N}$.  Consider the set ${\mathcal G}_+$ that consists of  $2\cdot 4^{2k}$ codes with positive center symbol,  (i.e. $A_+$ or $a_+$). Note, that ${\mathcal G}_+$ does not contain any $-{\bf R}$-symmetric element. For any triple  $(\tilde{\bf u},-\tilde{\bf u}, -{\bf R \tilde{u}})$  exactly one element belongs to ${\mathcal G}_+$. However, the set ${\mathcal G}_+$ contains $2\cdot 4^k$ ${\bf R}$-symmetric elements.  Nonsymmetric $2\cdot 4^{2k}-2\cdot 4^{k}$ elements of ${\mathcal G}_+$ can be split into pairs $(\tilde{\bf u},{\bf R \tilde{u}})$. Taking only one representative from each pair  $(\tilde{\bf u},{\bf R \tilde{u}})$ and returning $2\cdot 4^k$ ${\bf R}$-symmetric codes one concludes that the number of irreducible codes is
\begin{gather*}
\frac12\left(2\cdot 4^{2k}-2\cdot4^{k}\right)+2\cdot 4^k=16^k+4^k.
\end{gather*}

Let $N=2k$,  $k \in \mathbb{N}$. Consider the set ${\mathcal F}_+$ that consists of codes with positive first symbol. Evidently, ${\mathcal F}_+$ consists of $2\cdot 4^{2k-1}$ elements and for any pair  $(\tilde{\bf u},-\tilde{\bf u})$ exactly one element belongs to ${\mathcal F}_+$. The set ${\mathcal F}_+$ includes $2\cdot 4^{k-1}$ ${\bf R}$-symmetric codes and $2\cdot 4^{k-1}$ $-{\bf R}$-symmetric codes. Each vector $\tilde{\bf u}\in {\mathcal F}_+$ that is neither ${\bf R}$-symmetric nor $-{\bf R}$-symmetric has exactly one counterpart  in ${\mathcal F}_+$ that is either ${\bf R}$-symmetric or $-{\bf R}$-symmetric. This means that the number of irreducible codes is
\begin{gather*}
\frac12 (2\cdot 4^{2k-1}-2\cdot 4^{k-1}-2\cdot 4^{k-1})+2\cdot 4^{k-1}+2\cdot 4^{k-1}=\frac14 (16^k+2\cdot 4^k).
\end{gather*}
This completes the proof.
\end{proof}

\begin{example}
	\label{ex-2}
	For $N = 1$, the $2$ irreducible codes are given by 
	$$
	(a_+)  \;\; \mbox{\rm and} \;\; (A_+).
	$$
	For $N = 2$, the $6$ irreducible codes are given by 
	$$
	(a_+,a_+), \;\; (a_+,a_-), \;\; (a_+,A_+), \;\; (a_+,A_-), 
	\;\; (A_+,A_+), \;\; \mbox{\rm and} \;\; (A_+,A_-).
	$$
	The numbers of irreducible codes for $N = 1$ and $N = 2$ agree with Lemma \ref{lem-2}.
\end{example}

\begin{example}
	\label{ex-1}
	If $(a_+,a_-,A_+)$ is the code for $N = 3$, then 
	$$
(A_+,a_-,a_+), \;\; (a_-,a_+,A_-), \;\; \mbox{\rm and} \;\; 
(A_-,a_+,a_-)
$$ 
are equivalent codes. The code $(a_+,a_-,A_+)$ is one of the $18$ irreducible codes for $N = 3$. 
\end{example}

The existence result for the ILMs in the ACL is obtained by the implicit function theorem in the space of real-valued bi-infinite solutions $\mathbf{u} \in \mathbb{R}^{\mathbb{Z}}$ of the difference equation (\ref{Eq:u_gen_pq}) \cite{MA1994,A1997}. Moreover, it follows that the vector $\mathbf{u}$ is close to the limiting vector ${\bf u}^{(0)}$ for small $\varepsilon > 0$ and that the sign alternation of the vector $\tilde{\bf u} \in \mathbb{R}^N$ gives the sign alternation of the vector ${\bf u} \in \mathbb{R}^{\mathbb{Z}}$ for small $\varepsilon > 0$ \cite{ABK04,PKF05}. The existence result is given by the following proposition, which also yields the first assertion of Theorem \ref{th-main}.

\begin{prop}
	\label{prop-existence}
	Let $p,q \in \mathbb{N}$ be fixed such that $2 \leq p < q$. 
	For every $\gamma \in (0,\gamma_{{p,q}})$, there exists $\varepsilon_0 > 0$ and $C_0 > 0$ such that for every $\varepsilon \in (0,\varepsilon_0)$, there exists 
	a solution ${\bf u} \in \ell^2(\mathbb{Z})$ of the difference 
	equation (\ref{Eq:u_gen_pq}) such that 
\begin{equation}
\label{bound-u}
	\| {\bf u} - {\bf u}^{(0)} \|_{\ell^2(\mathbb{Z})} \leq C_0 \varepsilon,
\end{equation}
	where ${\bf u}^{(0)} = (\ldots,0,0,\tilde{\bf u},0,0,\ldots)$ with  $\tilde{\bf u} \in \mathbb{R}^N$ of any length $N \geq 1$. Moreover, 
	the number of sign alternations in ${\bf u}$ is equal to the number of sign alternations in $\tilde{\bf u}$.
\end{prop}

\begin{proof}
	By writing ${\bf u} = {\bf u}^{(0)} + {\bf v}$ with real-valued ${\bf v} \in \ell^2(\mathbb{Z})$, we rewrite the difference 
	equation (\ref{Eq:u_gen_pq}) in the equivalent form:
\begin{equation}
-\varepsilon \left(v_{n+1}-2v_n+v_{n-1}\right) + f'(u_n^{(0)}) v_n = \varepsilon \left( u_{n+1}^{(0)} - 2 u_n^{(0)} + u_{n-1}^{(0)} \right) - g(u_n^{(0)},v_n),
\label{Eq:v}
\end{equation}	
where $f(u) = u - |u|^{p-1} u + \gamma |u|^{q-1} u$ and $g(u,v) =  f(u + v) - f(u) - f'(u) v$. Since $p,q \in \mathbb{N}$ and $2 \leq p < q$, the mapping $v \to g(u,v)$ is $C^1$ for a given $u \in \mathbb{C}$ and 
there is $C > 0$ such that 
$$
|g(u_n^{(0)},v_n)| \leq C|v_n|^2, \quad \forall n \in \mathbb{Z},
$$
as long as $\|{\bf v} \|_{\ell^2} \leq 1$. We can rewrite (\ref{Eq:v}) in the abstract form:
\begin{equation}
\label{eq:v-abstract}
\mathcal{L} {\bf v} = {\bf H}({\bf v}), 
\end{equation}
with 
\begin{align*}
\mathcal{L} v_n &= f'(u_n^{(0)}) v_n - \varepsilon \Delta_2 v_n, \\
H_n({\bf v}) &= \varepsilon \Delta_2 u_n^{(0)} - g(u_n^{(0)},v_n).
\end{align*}
The mapping $\ell^2(\mathbb{Z}) \ni {\bf v} \to {\bf H}({\bf v}) \in \ell^2(\mathbb{Z})$ is $C^1$ since $\ell^2(\mathbb{Z})$ is a Banach algebra with respect to multiplication and
$$
\| \Delta_2 {\bf u}^{(0)} \|_{\ell^2} \leq 4 \|{\bf  u}^{(0)} \|_{\ell^2}. 
$$
On the other hand, we have $f'(0) = 1$, $f'(\pm a) < 0$, $f'(\pm A) > 0$ by Lemma \ref{lem-1}. Hence, $\mathcal{L}$ is an invertible operator in $\ell^2(\mathbb{Z})$ for sufficiently small $\varepsilon > 0$ with a bounded inverse. By the implicit function theorem in $\ell^2(\mathbb{Z})$, there exist $\varepsilon_0 > 0$ and $C_0 > 0$ such that there exists a unique solution 
to (\ref{eq:v-abstract})  satisfying $\| {\bf v} \|_{\ell^2} \leq C_0 \varepsilon$ for every $\varepsilon \in (0,\varepsilon_0)$. This yields the first assertion and results in the bound (\ref{bound-u}). 

It remains to prove preservation of the sign alternation in ${\bf u}$ from ${\bf u}^{(0)}$. To do so, we introduce $\Delta_2^-$ on $\ell^2(\mathbb{Z}_-)$ 
with $\mathbb{Z}_- = \{ \dots, -2,-1,0\}$ subject to the Dirichlet condition for $n = 1$. The difference equations (\ref{Eq:u_gen_pq}) for $n \in \mathbb{Z}_-$ can be rewritten in the form 
\begin{equation}
(1 - h(u_n) - \varepsilon \Delta_2^-) u_n = \varepsilon u_1 \delta_{n,0}, \quad n \in \mathbb{Z}_-,
\label{Eq:w}
\end{equation}
where $\delta_{n,0} = 1$ for $n = 0$ and $\delta_{n,0} = 0$ for $n \leq -1$, whereas $h(u) = |u|^{p-1} - \gamma |u|^{q-1}$. Since $p,q \in \mathbb{N}$ and $2 \leq p < q$, there is $C > 0$ such that 
$$
|h(u_n)| \leq C |u_n|, \quad \forall n \in \mathbb{Z}_-,
$$
as long as $\|{\bf u} \|_{\ell^2(\mathbb{Z}_-)} \leq 1$. Since $\|{\bf u} \|_{\ell^2(\mathbb{Z}_-)} = \mathcal{O}(\varepsilon)$ and $u_1 = \mathcal{O}(1)$ as $\varepsilon \to 0$ by the bound (\ref{bound-u}), it follows from (\ref{Eq:w}) that 
$$
u_n = \varepsilon^{|n|+1} u_1^{(0)} \left[ 1 + \mathcal{O}(\varepsilon) \right], \quad n \in \mathbb{Z}_-,
$$
so that the sign of $u_n$ for every $n \in \mathbb{Z}_-$ is the same as the sign of $u_1^{(0)}$ since $\varepsilon > 0$. Similarly, we get 
$$
u_n = \varepsilon^{n-N} u_N^{(0)} \left[ 1 + \mathcal{O}(\varepsilon) \right], \quad n \in \mathbb{Z}_+ = \{ N+1,N+2,\dots\}.
$$
Hence, the sign of $u_n$ for every $n \in \mathbb{Z}_+$ is the same as the sign of $u_N^{(0)}$ since $\varepsilon > 0$. This yields the equality between the number of sign alternations in ${\bf u}$ and that in $\tilde{\bf u}$.
\end{proof}

\section{Stability of ILMs in the anticontinuum limit}
\label{Sec:StabilityACL}

Let $\mathbf{u} \in \ell^2(\mathbb{Z})$ be the spatial profile of the ILM given by Proposition \ref{prop-existence}. Assuming that $\mathbf{u}$ is real-valued, 
we substitute ${\bf u} \to {\bf u} = {\bf v}(t) + i {\bf w}(t)$ with real-valued ${\bf v}(t), {\bf w}(t) \in \ell^2(\mathbb{Z})$ to the solution of the time-dependent DNLS equation (\ref{Eq:u_gen_pq_t}). Linearizing at the linear powers of ${\bf v}$, ${\bf w}$, we obtain the linearized DNLS equation in the form
\begin{align}
\frac{dv_n}{dt} + \varepsilon \Delta_2 w_n - w_n + |u_n|^{p-1} w_n - \gamma |u_n|^{q-1} w_n = 0, \label{linear-stab-1}\\
-\frac{dw_n}{dt} + \varepsilon \Delta_2 v_n - v_n + p |u_n|^{p-1} v_n - \gamma q  |u_n|^{q-1} v_n = 0. \label{linear-stab-2}
\end{align}
Separating the variables as ${\bf v}(t) = {\bf v} e^{i \omega t}$ and ${\bf w}(t) = i {\bf w} e^{i \omega t}$ in (\ref{linear-stab-1})--(\ref{linear-stab-2}), we obtain the eigenvalue problem for the time-independent eigenvector $(\mathbf{v},\mathbf{w}) \in \ell^2(\mathbb{Z}) \times \ell^2(\mathbb{Z})$ in the matrix form:
\begin{gather}
\omega\left(\begin{array}{c}
{\bf v}\\{\bf w}
\end{array}
\right)=
\left(\begin{array}{cc}
{\bf 0}& {\bf L}^-_\varepsilon\\
{\bf L}^+_\varepsilon& {\bf 0}
\end{array}
\right)
\left(\begin{array}{c}
{\bf v}\\{\bf w}
\end{array}
\right).
\label{spectral-stab}
\end{gather}
where the matrix operators ${\bf L}^-_\varepsilon$ and ${\bf L}^+_\varepsilon$ are given by 
\begin{gather*}
{\bf L}^-_\varepsilon=\left(
\begin{array}{ccccc}
\ddots&   \vdots    &   \vdots    &    \vdots    &  \vdots    \\[2mm]
\cdots&       D^-_{n-1} &     -\varepsilon        &   0             &    \cdots\\[2mm]
\cdots&                     -\varepsilon&    D^-_{n} &               -\varepsilon&    \ldots\\[2mm]
\cdots&                       0&              -\varepsilon&D^-_{n+1} &     \ldots\\[2mm]
\cdots&    \vdots   &       \vdots&         \vdots&        \ddots
\end{array}
\right),~
{\bf L}^+_\varepsilon=\left(
\begin{array}{ccccc}
\ddots&   \vdots    &   \vdots    &    \vdots    &  \vdots    \\[2mm]
\cdots&       D^+_{n-1} &     -\varepsilon        &   0             &    \cdots\\[2mm]
\cdots&                     -\varepsilon&    D^+_{n} &               -\varepsilon&    \ldots\\[2mm]
\cdots&                       0&              -\varepsilon&D^+_{n+1} &     \ldots\\[2mm]
\cdots&    \vdots   &       \vdots&         \vdots&        \ddots
\end{array}
\right).
\end{gather*}
with
\begin{gather*}
D_n^-=2\varepsilon+1-|u_n|^{p-1}+\gamma |u_n|^{q-1},\quad D_n^+=2\varepsilon+1-p|u_n|^{p-1}+\gamma q |u_n|^{q-1}.
\end{gather*}
The ILM with the spatial profile $\mathbf{u}$ is called spectrally stable 
if ${\rm Im}~\omega=0$ for all eigenvalues $\omega$. The eigenvalue problem (\ref{spectral-stab}) can be rewritten in the equivalent scalar forms
\begin{gather}
{\bf L}^-_\varepsilon{\bf L}^+_\varepsilon{\bf v} = \lambda  {\bf v},   
\quad {\bf L}^+_\varepsilon{\bf L}^-_\varepsilon{\bf w} = \lambda  {\bf w},
\label{Eq:Eigen02}
\end{gather}
where $\lambda \equiv \omega^2$. If the linear operator ${\bf L}^+_\varepsilon : \ell^2(\mathbb{Z}) \to \ell^2(\mathbb{Z})$ is invertible, e.g. for small $\varepsilon > 0$ as in the proof of Proposition \ref{prop-existence}, 
then the second eigenvalue problem in (\ref{Eq:Eigen02}) can be rewritten 
as the generalized eigenvalue problem for the eigenvector ${\bf w} \in \ell^2(\mathbb{Z})$:
\begin{gather}
{\bf L}^-_\varepsilon{\bf w} = \lambda ({\bf L}^+_\varepsilon)^{-1} {\bf w}.
\label{Gener}
\end{gather}
Since $\lambda = \omega^2$, the neutrally stable eigenvalues $\omega \in \mathbb{R}$ correspond to positive eigenvalues $\lambda > 0$ of the generalized eigenvalue problem (\ref{Gener}), whereas the unstable eigenvalues $\omega \notin \mathbb{R}$ correspond to either negative eigenvalues $\lambda < 0$ or complex eigenvalues $\lambda \notin \mathbb{R}$.

\begin{remark}
	The eigenvalue problem (\ref{spectral-stab}) contains the same spectrum 
	if $\mathbf{u}$ is replaced ${\bf R} \mathbf{u}$, $-{\mathbf{u}}$, and $-{\bf R} \mathbf{u}$, where ${\bf R}$ is the reversibility operator given by $({\bf R u})_n = u_{-n}$. This further suggests that the spectral stability of the four equivalent ILMs is identical, so that we can consider only one of the four equivalent ILMs.
\end{remark}

\subsection{Truncated eigenvalue problem}
\label{Sec:StabilityACLGen}

By Proposition \ref{prop-existence}, the spatial profile ${\bf u} \in \ell^2(\mathbb{Z})$ is close to the limiting profile ${\bf u}^{(0)} = (\ldots,0,0,\tilde{\bf u},0,0,\ldots)$ with $\tilde{\bf u} \in \mathbb{R}^N$ 
of any length $N \geq 1$. As is explained in Section \ref{Sect:StatStates}, we only consider the nonzero elements in $\tilde{\bf u}$ located between $n = 1$ and $n = N$.

If $\varepsilon=0$, then ${\bf L}^-_0$ and ${\bf L}^+_0$ are diagonal matrix operators such as ${\bf L}^-_0$ has a block of $N$ zero diagonal elements with all other diagonal elements at $1$ and  ${\bf L}^+_0$ has a block of $N$ diagonal elements $(\tilde{D}_1^+,\tilde{D}_2^+,\ldots,\tilde{D}_N^+)$, with all other diagonal elements at $1$, where 
\begin{gather}
\tilde{D}_n^+=\left[
\begin{array}{ll}
f'(a),& \quad \mbox{if}~\tilde{u}_n=\pm a,\\ [2mm]
f'(A),&\quad \mbox{if}~\tilde{u}_n=\pm A,
\end{array}
\right.
\label{Eq:D_n^+}
\end{gather}
with $f'(u) \equiv 1- p u^{p-1} +\gamma q u^{q-1}$ for $u \in (0,\infty)$.

To get a nontrivial truncated eigenvalue problem as $\varepsilon \to 0$, we 
note that $\tilde{u}_n \neq 0$ for $n = 1,\ldots,N$. By using Eq.~(\ref{Eq:u_gen_pq}), we can write 
$$
D_n^- = 2 \varepsilon + 1 - |u_n|^{p-1} + \gamma |u_n|^{q-1} = \varepsilon \frac{u_{n-1} + u_{n+1}}{u_n}, \quad n = 1, \dots, N.
$$
Truncating at the leading order with $u_n^{(0)} = \tilde{u}_n$ for $n = 1,\dots,N$, we introduce 
\begin{gather}
\label{Eq:D_n^-}
\tilde{D}_n^-=\frac{\tilde{u}_{n-1}+\tilde{u}_{n+1}}{\tilde{u}_n}, \quad n=1,\ldots N, 
\end{gather}
subject to the boundary conditions $\tilde{u}_0=\tilde{u}_{N+1}=0$.
Truncation of the generalized eigenvalue problem  (\ref{Gener}) at $n = 1,\ldots,N$  yields the following problem at the leading order:
\begin{gather}
\tilde{\bf L}^-\tilde{\bf w} = \tilde{\lambda}  (\tilde{\bf L}^+)^{-1} \tilde{\bf w}. \label{Gener-matrix}
\end{gather}
where $N\times N$ matrices $\tilde{\bf L}^-$ and $\tilde{\bf L}^+$ are given by 
\begin{gather*}
\tilde{\bf L}^-=\left(
\begin{array}{ccccc}
\tilde{D}_1^-&   -1  &   0   &   \cdots    & 0   \\[2mm]
-1&      \tilde{D}_2^-& -1  &              0  &    \vdots\\[2mm]
0&            -1&     \ddots&         -1&    0\\[2mm]
\vdots&             0&               -1&   \ddots  & -1\\[2mm]
0&   \cdots          &         0&                  -1&     \tilde{D}_N^-
\end{array}
\right),~
\tilde{\bf L}^+=\left(
\begin{array}{ccccc}
\tilde{D}_1^+&   0  &   0   &   \cdots    & 0   \\[2mm]
0&      \tilde{D}_2^+& 0  &              0  &    \vdots\\[2mm]
0&           0&     \ddots&        0&    0\\[2mm]
\vdots&             0&              0&   \ddots  &0\\[2mm]
0&   \cdots          &         0&                0&     \tilde{D}_N^+
\end{array}
\right),
\end{gather*}
with $\tilde{D}_n^+$ and $\tilde{D}_n^-$ given by (\ref{Eq:D_n^+}) and (\ref{Eq:D_n^-}), respectively. The following lemma gives the relation 
between eigenvalues of (\ref{Gener}) and (\ref{Gener-matrix}).

\begin{lemma}
	\label{theorem-1}
	Let $\{\tilde{\lambda}_1, \dots, \tilde{\lambda}_N \}$ be eigenvalues of the truncated eigenvalue problem  (\ref{Gener-matrix}). If $\tilde{\lambda}_k$ for some $k \in \{1,2,\ldots,N\}$ is simple, then the generalized eigenvalue problem (\ref{Gener}) admits a simple eigenvalue $\lambda_k$ given by 
	\begin{gather}
	\label{eigen-expansion}
	\lambda_k = \varepsilon\tilde{\lambda}_k + \mathcal{O}(\varepsilon^2).
	\end{gather}
Moreover, if $\tilde{\lambda}_k \in \mathbb{R}$, then $\lambda_k \in \mathbb{R}$.
\end{lemma} 

\begin{proof}
We have 
	$$
\Pi_{N \times N} {\bf L}_{\varepsilon}^- = \varepsilon \tilde{\bf L}^- + \mathcal{O}(\varepsilon^2), \quad 	
\Pi_{N \times N} {\bf L}_{\varepsilon}^+ = \tilde{\bf L}^+ + \mathcal{O}(\varepsilon),
	$$
where $\Pi_{N \times N}$ is the truncation of the matrix operator 
defined in $\ell^2(\mathbb{Z})$ to the matrix in $\mathbb{M}^{N \times N}$ 
by using the rows and columns between $n = 1$ and $n = N$. 
In the limit $\varepsilon \to 0$, we obtain by Proposition \ref{prop-existence} that
$$
{\bf L}_{\varepsilon = 0}^- = {\rm diag}({\bf I}_{\leq 0},{\bf 0}_{N \times N},{\bf I}_{\geq N+1}), \qquad {\bf L}_{\varepsilon = 0}^+= {\rm diag}({\bf I}_{\leq 0},\tilde{\bf L}^+,{\bf I}_{\geq N+1}),
$$
where ${\bf I}_{\leq 0}$ and ${\bf I}_{\geq N+1}$ are identity operators for $n \in \mathbb{Z}_-$ and $n \in \mathbb{Z}_+$ defined in the proof of Proposition \ref{prop-existence}, whereas ${\bf 0}_{N \times N}$ is the zero matrix for $1 \leq n \leq N$. It is clear that the spectrum of the generalized eigenvalue 
problem (\ref{Gener}) consists of only two eigenvalues: $\lambda = 0$ of multiplicity $N$ in the subspace $U_N = {\rm span}({\bf e}_1,{\bf e}_2,\dots,{\bf e}_N) \subset \ell^2(\mathbb{Z})$ and 
$\lambda = 1$ of infinite multiplicity in the subspaces 
$\ell^2(\mathbb{Z}_-)$ and $\ell^2(\mathbb{Z}_+)$. Since 
$$
\ell^2(\mathbb{Z}) = \ell^2(\mathbb{Z}_-) \oplus U_N \oplus \ell^2(\mathbb{Z}_+),
$$
we can proceed with the perturbation theory for the zero eigenvalue and use the decomposition 
$$
{\bf w} = c_1 {\bf e}_1 + c_2 {\bf e}_2 + \dots c_N {\bf e}_N + {\bf w}_- + {\bf w}_+, 
$$
where ${\bf w}_- \in \ell^2(\mathbb{Z}_-)$ and ${\bf w}_+ \in \ell(\mathbb{Z}_+)$. By the implicit function theorem (similar to the one used in the proof of Proposition \ref{prop-existence}), there exist $\varepsilon_0 > 0$, $C_0 > 0$, and $\lambda_0 > 0$ such that for every $\varepsilon \in (0,\varepsilon_0)$ and $|\lambda| < \lambda_0$, there exists a unique mapping 
$\mathbb{C}^N \ni {\bf c}=(c_1,c_2,\dots,c_N) \to ({\bf w}_-,{\bf w}_+) \in \ell(\mathbb{Z}_-) \oplus \ell(\mathbb{Z}_+)$ satisfying
$$
\| {\bf w}_- \|_{\ell^2(\mathbb{Z}_-)} + \| {\bf w}_+ \|_{\ell^2(\mathbb{Z}_+)} \leq C_0 \varepsilon \| {\bf c} \|_{\mathbb{C}^N}. 
$$
Rescaling the eigenvalues as $\lambda = \varepsilon \tilde{\lambda}(\varepsilon)$ and using the unique mapping above, we get 
the matrix eigenvalue problem in the form 
\begin{equation}
\label{perturbed-eig}
\left[ \tilde{\bf L}^- + \mathcal{O}(\varepsilon) \right] {\bf c} = 
\tilde{\lambda}(\varepsilon) \left[ \tilde{\bf L}^+ + \mathcal{O}(\varepsilon) \right]^{-1} {\bf c}.
\end{equation}
The truncated version of the matrix eigenvalue problem (\ref{perturbed-eig}) coincides with (\ref{Gener-matrix}) and admits eigenvalues $\{\tilde{\lambda}_1, \dots, \tilde{\lambda}_N \}$. Simple roots of the characteristic polynomial for the matrix eigenvalue problem (\ref{perturbed-eig}) persist with respect to $\mathcal{O}(\varepsilon)$ perturbation terms and yield the expansion (\ref{eigen-expansion}). This is again proven with the implicit function 
theorem for roots of the characteristic polynomial. Furthermore, since complex eigenvalues are symmetric about real axis, that is, both $\lambda$ and $\bar{\lambda}$ are eigenvalues, then if $\tilde{\lambda}_k$ is simple and real, then $\lambda_k$ is simple and real for small $\varepsilon > 0$, as $\lambda_k$ cannot split into a pair of two complex eigenvalues due to preservation of multiplicity in the roots of the characteristic polynomial with respect to $\mathcal{O}(\varepsilon)$ perturbation terms. 
\end{proof}

\begin{remark}
	If $\tilde{\lambda}_k$ is a multiple eigenvalue, then $\lambda_k$ is still located near $\tilde{\lambda}_k$ due to Puiseux expansions. However, if the algebraic multiplicity exceeds the geometric multiplicity of $\tilde{\lambda}_k$, the correction terms may not be real-valued even if $\tilde{\lambda}_k \in \mathbb{R}$.
\end{remark}

\subsection{Stability results based on the truncated eigenvalue problem}

Eigenvalues of the truncated eigenvalue problem (\ref{Gener-matrix}) 
are real if either $\tilde{\bf L}^-$ or $\tilde{\bf L}^+$ is positive definite. 
Therefore, we first study the number of negative and positive eigenvalues 
in these matrices, denoted as $n(\tilde{\bf L}^{\pm})$ and $p(\tilde{\bf L}^{\pm})$ respectively.

For $\tilde{\bf L}^+$, let $K$ be the number of elements $\pm a$ in the code $\tilde{\bf u}$. By Lemma \ref{lem-1}, we have 
\begin{equation}
\label{neg-2}
n(\tilde{\bf L}^+) = K, \quad p(\tilde{\bf L}^+) = N-K.
\end{equation}  
For $\tilde{\bf L}^-$, let $N_0({\bf z})$ be the number of flips (changes of sign between nearest components) of ${\bf z}$ for either ${\bf z} \in \mathbb{R}^N$ or ${\bf z} \in \ell^2(\mathbb{Z})$. The following 
lemma characterizes $n(\tilde{\bf L}^-)$ and $p(\tilde{\bf L}^-)$. 

\begin{lemma}
	\label{theorem-2}
$\tilde{\bf L}^-$ admits a simple zero eigenvalue, whereas
	\begin{equation}
	\label{neg-1}
	n(\tilde{\bf L}^-) = N_0(\tilde{\bf u}), \quad p(\tilde{\bf L}^-) = N-N_0(\tilde{\bf u})-1. 
	\end{equation}
Similarly, the spectrum of operator ${\bf L}^-_{\varepsilon} : \ell^2(\mathbb{Z}) \to \ell^2(\mathbb{Z})$ includes a simple zero eigenvalue, $N_0(\tilde{\bf u})$ negative eigenvalues, and the rest of its spectrum is strictly positive and bounded away from zero for $\varepsilon > 0$ by $C_0 \varepsilon$ with $C_0 > 0$.
\end{lemma}

\begin{proof}
	The spectral problems for both $\tilde{\bf L}^-$ and ${\bf L}^-_{\varepsilon}$ belong to the class of discrete Schr\"{o}dinger equations:
	$$
-\Delta_2 v_n + (\tilde{D}_n^- - 2) v_n = \lambda v_n, \quad n = 1,\dots,N, 
\quad v_0 = v_{N+1} = 0
	$$
	and
	$$
	-\varepsilon \Delta_2 v_n + v_n - |u_n|^{p-1} v_n + \gamma |u_n|^{q-1} v_n = \lambda v_n, \quad n \in \mathbb{Z},
	$$
	where $\lambda$ is the spectral parameter. It follows from 
	(\ref{Eq:D_n^-}) that $\tilde{\bf L}^-\tilde{\bf u}={\bf 0}$. 
	Furthermore, it follows from (\ref{Eq:u_gen_pq}) that 
	${\bf L}^-_{\varepsilon} {\bf u}={\bf 0}$. By Proposition \ref{prop-existence}, we have $N_0({\bf u})= N_0(\tilde{\bf u})$. 
	Sturm's comparison theorem for the discrete Schr\"{o}dinger equations \cite{LL} 
	states that the number of negative eigenvalues of either $\tilde{\bf L}^-$ or ${\bf L}^-_{\varepsilon}$ is equal to the number of sign alternation 
	in the eigenvector for the zero eigenvalue, either $\tilde{\bf u}$ or ${\bf u}$. This yields the result for $n(\tilde{\bf L}^-) = n({\bf L}^-_{\varepsilon})$ for small $\varepsilon > 0$. Since $\tilde{\bf L}^-$ is the $N\times N$ matrix with a simple zero eigenvalue, we have 
	$p(\tilde{\bf L}^-)= N - 1 -  n(\tilde{\bf L}^-)$. Since ${\bf L}^-_{\varepsilon}$ is a linear operator in $\ell^2(\mathbb{Z})$ 
	satisfying ${\bf L}_{\varepsilon = 0}^- = {\rm diag}({\bf I}_{\leq 0},{\bf 0}_{N \times N},{\bf I}_{\geq N+1})$, its spectrum consists of 
	$n({\bf L}^-_{\varepsilon})$ negative eigenvalues, a simple zero eigenvalue, and $N-1-n({\bf L}^-_{\varepsilon})$ positive eigenvalues of the $\mathcal{O}(\varepsilon)$ order, as well as the rest of the spectrum which is located near $1$ for small $\varepsilon > 0$. Hence, the strictly positive spectrum is bounded away from zero for $\varepsilon > 0$ by $C_0 \varepsilon$ with $C_0 > 0$.
\end{proof}

The following four propositions present definite results on 
the count of real eigenvalues in the truncated eigenvalue problem (\ref{Gener-matrix}) in the cases when either $\tilde{\bf L}^-$ or $\tilde{\bf L}^+$ is positive definite. If there exist negative eigenvalues $\lambda$, then the result of Lemma \ref{theorem-1} implies that the ILM is spectrally unstable in the ACL. If no negative eigenvalues exist, we are able to conclude on the 
spectral stability of the corresponding ILMs in the ACL. This yields assertions (A), (B), (C), and (D) in Theorem \ref{th-main}.

\begin{prop}
	\label{prop-1}
	If the code $\mathcal{A}$ includes only symbols $A_+$ and $A_-$, then 
	the truncated eigenvalue problem (\ref{Gener-matrix}) admits $N_0(\tilde{\bf u})$ negative eigenvalues, a simple zero eigenvalue, and $N - N_0(\tilde{\bf u}) -1$ positive eigenvalues. 
\end{prop}

\begin{proof}
	It follows from (\ref{neg-2}) with $K = 0$ that the matrix $\tilde{\bf L}^+$ is positive-definite. In this case, the generalized eigenvalue problem (\ref{Gener-matrix}) only admits real eigenvalues and Sylvester's inertial theorem counts the number of negative, zero, and positive eigenvalues $\tilde{\lambda}$ from those of $\tilde{\bf L}^-$ \cite{Gelfand}. This yields the result by the count (\ref{neg-1}) in Lemma \ref{theorem-2}.
\end{proof}

\begin{remark}
	The spectrally stable codes correspond to $N_0(\tilde{\bf u}) = 0$, which is realized at the nonalternating symbols, e.g. $(A_+)$, $(A_-)$, $(A_+A_+)$, $(A_-A_-)$, $(A_+A_+A_+)$, $(A_-A_-A_-)$ etc. Since the quadratic forms 
	$\tilde{\bf L}^- \tilde{\bf w}  \cdot \tilde{\bf w} $ and  $(\tilde{\bf L}^+)^{-1} \tilde{\bf w}  \cdot \tilde{\bf w} $ are strictly positive for a positive eigenvalue $\tilde{\lambda}$ with the eigenvector $\tilde{\bf w} \neq {\bf 0}$, then the positive eigenvalues of (\ref{Gener-matrix}) are semi-simple. Hence, they persist as real positive eigenvalues of (\ref{Gener}) according to Lemma \ref{theorem-1}. This yields the spectral stability of the corresponding ILMs. All other codes which include only symbols $A_+$ and $A_-$ are spectrally unstable. 
\end{remark}

\begin{prop}
	\label{prop-2}
	If the code $\mathcal{A}$ includes only symbols $a_+$ and $a_-$, then 
	the truncated eigenvalue problem (\ref{Gener-matrix}) admits $N - N_0(\tilde{\bf u}) -1$ negative eigenvalues, a simple zero eigenvalue, and $N_0(\tilde{\bf u})$ positive eigenvalues. 
\end{prop}

\begin{proof}
It follows from (\ref{neg-2}) with $K = N$ that the matrix $\tilde{\bf L}^+$ is negative-definite. In this case, the generalized eigenvalue problem (\ref{Gener-matrix}) only admits real eigenvalues and Sylvester's inertial theorem counts the number of negative, zero, and positive eigenvalues $\tilde{\lambda}$ from those for $-\tilde{\bf L}^-$ \cite{Gelfand}. This yields the result by the count (\ref{neg-1}) in Lemma \ref{theorem-2}.
\end{proof}

\begin{remark}
	The spectrally stable codes correspond to $N = N_0(\tilde{\bf u}) + 1$, 
	which is realized at the alternating combination of symbols $a_+$ and $a_-$, e.g. $(a_+)$, $(a_-)$, $(a_+a_-)$, $(a_+a_-a_+)$, $(a_-a_+a_-)$, etc. Since the quadratic forms $\tilde{\bf L}^- \tilde{\bf w}  \cdot \tilde{\bf w} $ and  $(\tilde{\bf L}^+)^{-1} \tilde{\bf w}  \cdot \tilde{\bf w} $ are strictly negative for a positive eigenvalue $\tilde{\lambda}$ with the eigenvector $\tilde{\bf w} \neq {\bf 0}$, then the positive eigenvalues of (\ref{Gener-matrix}) are semi-simple. Hence, they persist as real positive eigenvalues of (\ref{Gener}) according to Lemma \ref{theorem-1}. This yields the spectral stability of the corresponding ILMs. 
	All other codes which include only symbols $a_+$ and $a_-$ are spectrally unstable. 
\end{remark}

\begin{prop}
	\label{prop-3}
	Assume that $(\tilde{\bf L}^+)^{-1} \tilde{\bf u} \cdot \tilde{\bf u} \neq 0$ and define 
	$$
	\sigma_0 = \left\{ \begin{array}{ll} 1, \quad & \mbox{\rm if} \quad (\tilde{\bf L}^+)^{-1} \tilde{\bf u} \cdot \tilde{\bf u} < 0, \\
	0, \quad & \mbox{\rm if} \quad  (\tilde{\bf L}^+)^{-1} \tilde{\bf u} \cdot \tilde{\bf u} > 0. \end{array} \right.
	$$
	If the code $\mathcal{A}$ includes only symbols $a_+$ and $A_+$, then 
	the truncated eigenvalue problem (\ref{Gener-matrix}) admits $K - \sigma_0$ negative eigenvalues, a simple zero eigenvalue, and $N - K - (1-\sigma_0)$ positive eigenvalues.
\end{prop}

\begin{proof}
	By Lemma \ref{theorem-2}, the matrix $\tilde{\bf L}^-$ has a simple zero eigenvalue with the eigenvector $\tilde{\bf u}$ and $(N-1)$ positive eigenvalues since $N_0(\tilde{\bf u}) = 0$. By Fredholm's theorem, the zero eigenvalue in the generalized eigenvalue problem (\ref{Gener-matrix}) is simple if 
	$$
	(\tilde{\bf L}^+)^{-1} \tilde{\bf u} \cdot \tilde{\bf u} \neq 0
	$$
	Consider the orthogonal complement of $\tilde{\bf u}$ in $\mathbb{R}^N$ 
with the orthogonal projection operator $\Pi_0 : \mathbb{R}^N \to \mathbb{R}^N |_{\tilde{\bf u}^{\perp}}$. The generalized eigenvalue problem (\ref{Gener-matrix}) can be reduced on $\mathbb{R}^N |_{\tilde{\bf u}^{\perp}}$ by using 
\begin{equation}
\Pi_0 \tilde{\bf L}^- \Pi_0 \tilde{\bf w} = \tilde{\lambda}  \Pi_0 (\tilde{\bf L}^+)^{-1} \Pi_0 \tilde{\bf w}. 
\label{Gener-projected}
\end{equation}
Since $\Pi_0 \tilde{\bf L}^- \Pi_0$ is strictly positive, Sylvester's inertial theorem counts the number of remaining $(N-1)$ negative and positive eigenvalues $\tilde{\lambda}$ in (\ref{Gener-projected}) from those for $\Pi_0 (\tilde{\bf L}^+)^{-1} \Pi_0 $. It follows from (\ref{neg-2}) that 
$(\tilde{\bf L}^+)^{-1}$ has $K$ negative and $N-K$ positive eigenvalues. 
If $(\tilde{\bf L}^+)^{-1} \tilde{\bf u} \cdot \tilde{\bf u} < 0$, then 
$\Pi_0 (\tilde{\bf L}^+)^{-1} \Pi_0 $ has $K-1$ negative and $N-K$ positive 
eigenvalues. If $(\tilde{\bf L}^+)^{-1} \tilde{\bf u} \cdot \tilde{\bf u} > 0$, then $\Pi_0 (\tilde{\bf L}^+)^{-1} \Pi_0 $ has $K$ negative and $N-K-1$ positive 
eigenvalues. This yields the result. 
\end{proof}

\begin{remark}
	The spectrally stable codes correspond to $K = \sigma_0$, which can be realized in two different ways. Either no symbols $a_+$ are present if $\sigma_0 = 0$, in which case the assertion follows by Proposition \ref{prop-1}, or only one symbol $a_+$ is present if $\sigma_0 = 1$. In either case, the quadratic forms $\tilde{\bf L}^- \tilde{\bf w}  \cdot \tilde{\bf w} $ and  $(\tilde{\bf L}^+)^{-1} \tilde{\bf w}  \cdot \tilde{\bf w} $ are strictly positive for a positive eigenvalue $\tilde{\lambda}$ with the eigenvector $\tilde{\bf w} \neq {\bf 0}$. Hence, the positive eigenvalues of (\ref{Gener-projected}) are semi-simple and persist as real positive eigenvalues of (\ref{Gener}). This yields the spectral stability of the corresponding ILMs. All other codes which include only sumbols $a_+$ and $A_+$ are spectrally unstable.  In particular, any code of this type that includes more than one symbol $a_+$ is spectrally unstable.
\end{remark}

\begin{prop}
	\label{prop-4}
	Assume that $(\tilde{\bf L}^+)^{-1} \tilde{\bf u} \cdot \tilde{\bf u} \neq 0$ and define $\sigma_0$ as in Proposition \ref{prop-3}. 
	If the code $\mathcal{A}$ includes only sign-alternating symbols $a_\pm$ and $A_\pm$, e.g. $(A_+ a_-a_+A_-)$,  then 	the truncated eigenvalue problem (\ref{Gener-matrix}) admits $N - K - (1-\sigma_0)$ negative eigenvalues, a simple zero eigenvalue, and $K - \sigma_0$ positive eigenvalues.
\end{prop}

\begin{proof}
	By Lemma \ref{theorem-2}, the matrix $\tilde{\bf L}^-$ has a simple zero eigenvalue with the eigenvector $\tilde{\bf u}$ and $(N-1)$ negative eigenvalues since $N_0(\tilde{\bf u}) = N-1$. By Fredholm's theorem, the zero eigenvalue in the generalized eigenvalue problem (\ref{Gener-matrix}) is simple if 
	$$
	(\tilde{\bf L}^+)^{-1} \tilde{\bf u} \cdot \tilde{\bf u} \neq 0
	$$ 
	By using the same reduction to (\ref{Gener-projected}), 
	we obtain a strictly negative $\Pi_0 \tilde{\bf L}^- \Pi_0$ so that 
	Sylvester's inertial theorem counts the number of remaining $(N-1)$ negative and positive eigenvalues $\tilde{\lambda}$ in (\ref{Gener-projected}) from 
	those of $-\Pi_0 (\tilde{\bf L}^+)^{-1} \Pi_0$. By using the same 
	argument as in the proof of Proposition \ref{prop-3}, this yields the result.
\end{proof}

\begin{remark}
	The spectrally stable codes correspond to $N = K + 1 - \sigma_0$, which can be realized in two different ways. Either no symbols $A_{\pm}$ are present if $\sigma_0 = 1$, in which case the assertion follows by Proposition \ref{prop-2}, or only one symbol $A_{\pm}$ is present if $\sigma_0 = 0$. In either case, the quadratic forms 
	$\tilde{\bf L}^- \tilde{\bf w}  \cdot \tilde{\bf w} $ and  $(\tilde{\bf L}^+)^{-1} \tilde{\bf w}  \cdot \tilde{\bf w} $ are strictly negative for a positive eigenvalue $\tilde{\lambda}$ with the eigenvector $\tilde{\bf w} \neq {\bf 0}$. Hence, the positive eigenvalues of (\ref{Gener-projected}) are semi-simple and persist as real positive eigenvalues of (\ref{Gener}). This yields the spectral stability of the corresponding ILMs. All other codes which include only sign-alternating symbols $a_{\pm}$ and $A_{\pm}$ are spectrally unstable. 
\end{remark}

\subsection{Stability results based on the negative index theory}
\label{Sec:NegativeIndex}

To give a general count of unstable eigenvalues in the generalized eigenvalue problem (\ref{Gener}) for small $\varepsilon > 0$, let us first clarify 
the quantity $ (\tilde{\bf L}^+)^{-1} \tilde{\bf u} \cdot \tilde{\bf u} $ 
which appears in Propositions \ref{prop-3} and \ref{prop-4}. By using 
the definition $\tilde{D}_n^+$ in (\ref{Eq:D_n^+}), we derive 
\begin{equation}
\label{neg-3}
(\tilde{\bf L}^+)^{-1} \tilde{\bf u} \cdot \tilde{\bf u} = \frac{K a^2}{f'(a)} + \frac{(N-K) A^2}{f'(A)},
\end{equation}
where $f'(a) = 1 - p a^{p-1} + \gamma q a^{q-1} < 0$ and $f'(A) = 1 - p A^{p-1} + \gamma q A^{q-1} > 0$ by Lemma \ref{lem-1}. By the continuity of ${\bf u}$ with respect to $\varepsilon > 0$ in Proposition \ref{prop-existence}, we also 
get 
\begin{equation}
\label{neg-4}
\lim_{\varepsilon \to 0} \langle (\bf{L}^+_{\varepsilon})^{-1} {\bf u}, {\bf u} \rangle = (\tilde{\bf L}^+)^{-1} \tilde{\bf u} \cdot \tilde{\bf u}, 
\end{equation}
where $\langle \cdot, \cdot \rangle$ denotes the standard inner product in $\ell^2(\mathbb{Z})$. 

We recall that the ILM with the spatial profile ${\bf u}$ is a critical point 
of the augmented energy functional $\Lambda({\bf u}) = H({\bf u}) + Q({\bf u})$ in (\ref{aug-energy}), where the energy $H({\bf u})$ and the mass $Q({\bf u})$ are given by (\ref{energy}) and (\ref{mass}). It is straightforward to verify that the critical points ${\bf u} \in \ell^2(\mathbb{Z})$ of the augmented energy $\Lambda$ satisfy the difference equation (\ref{Eq:u_gen_pq}). If ${\bf u}$ is real-valued, the quadratic form of $\Lambda$ at ${\bf u}$ with the perturbation ${\bf v} + i {\bf w}$ for real-valued ${\bf v}, {\bf w} \in \ell^2(\mathbb{Z})$ is given by the sum of two quadratic forms 
\begin{equation*}
\Lambda({\bf u} + {\bf v} + i {\bf w}) - \Lambda({\bf u}) = \langle {\bf L}^+_{\varepsilon} {\bf v}, {\bf v} \rangle + \langle {\bf L}^-_{\varepsilon} {\bf w}, {\bf w} \rangle + \mathcal{O}(\| {\bf v} + i {\bf w} \|^3_{\ell^2}),
\end{equation*}
where ${\bf L}^{\pm}_{\varepsilon}$ are the same operators as in the spectral stability problem (\ref{spectral-stab}). By using the stability theory in Hamiltonian systems (see Theorems 1.8 and 3.2 in \cite{GP25}), we give a general count of unstable eigenvalues 
in the generalized eigenvalue problem (\ref{Gener}). 

\begin{prop}
	\label{theorem-count}
	Assume that ${\rm Ker}({\bf L}^+_{\varepsilon}) = \{ 0 \}$, 
	${\rm Ker}({\bf L}^-_{\varepsilon}) = {\rm span}\{{\bf u}\}$, 
$\langle ({\bf L}_{\varepsilon}^+)^{-1} {\bf u}, {\bf u} \rangle \neq 0$, and define 
$$
\sigma = \left\{ \begin{array}{ll} 1, \quad & \mbox{\rm if} \quad \langle ({\bf L}_{\varepsilon}^+)^{-1} {\bf u}, {\bf u} \rangle  < 0, \\
0, \quad & \mbox{\rm if} \quad  \langle ({\bf L}_{\varepsilon}^+)^{-1} {\bf u}, {\bf u} \rangle  > 0. \end{array} \right.
$$ 
Referring to the generalized eigenvalue problem (\ref{Gener}), define the number $N_c$ of complex eigenvalues $\lambda$  with ${\rm Im}(\lambda) > 0$, the number $N_r^{+}$ ($N_r^-$) of real negative eigenvalues $\lambda < 0$ with eigenvectors ${\bf w} \in \ell^2(\mathbb{Z})$ having positive (negative) values of $\langle ({\bf L}_{\varepsilon}^+)^{-1} {\bf w}, {\bf w} \rangle$, and the number $N_i^{+}$ ($N_i^-$) of real positive eigenvalues $\lambda > 0$ with eigenvectors ${\bf w} \in \ell^2(\mathbb{Z})$ having positive (negative) values of $\langle ({\bf L}_{\varepsilon}^+)^{-1} {\bf w}, {\bf w} \rangle$. 
The numbers of eigenvalues satisfy the following completeness relations:
	\begin{align}
	\label{count-1}
	N_c + N_r^- + N_i^- &= n({\bf L}_{\varepsilon}^+) - \sigma, \\
	\label{count-2}
	N_c + N_r^+ + N_i^- &= n({\bf L}_{\varepsilon}^-),
	\end{align}
where $n({\bf L}_{\varepsilon}^{\pm})$ denotes the number of negative eigenvalues of ${\bf L}_{\varepsilon}^{\pm}$ in $\ell^2(\mathbb{Z})$.
\end{prop}

The following remarks clarify the statement of Proposition \ref{theorem-count}. 

\begin{remark}
	The assumptions of Proposition \ref{theorem-count} are satisfied 
	for small $\varepsilon > 0$ due to invertibility of $\tilde{\bf L}^+$, 
	a simple zero eigenvalue of $\tilde{\bf L}^-$, and the continuity result (\ref{neg-4}) if $\langle (\tilde{\bf L}^+)^{-1} \tilde{\bf u}, \tilde{\bf u} \rangle \neq 0$. Hence, $\sigma = \sigma_0$, where $\sigma_0$ appears in Propositions \ref{prop-3} and \ref{prop-4}.
\end{remark}

\begin{remark}
	If $n({\bf L}_{\varepsilon}^+) = n({\bf L}_{\varepsilon}^-) = 0$, 
	then necessarily $\sigma_0 = 0$ and the ILM is a local minimizer of the augmented energy $\Lambda$ which is degenerate only due to the rotational symmetry represented by ${\rm Ker}({\bf L}^-_{\varepsilon}) = {\rm span}\{{\bf u}\}$. Since $N_c = N_r^{\pm} = N_i^- = 0$ follow from (\ref{count-1})--(\ref{count-2}), the local minimizer is spectrally stable. It is also nonlinearly (orbitally) stable (see Theorem 2.3 in \cite{GP25}).
\end{remark}

\begin{remark}
	If $n({\bf L}_{\varepsilon}^+) = 1$, $n({\bf L}_{\varepsilon}^-) = 0$, 
	then the ILM is a saddle point of the augmented energy $\Lambda$. 
	If $\sigma = 1$, it is a local minimizer of energy $H$ for fixed mass $Q$, which is degenerate only due to the rotational symmetry represented by ${\rm Ker}({\bf L}^-_{\varepsilon}) = {\rm span}\{{\bf u}\}$. In this context, the quantity $\langle ({\bf L}_{\varepsilon}^+)^{-1} {\bf u}, {\bf u} \rangle$ is related to the Vakhitov--Kolokolov slope condition 
	\begin{equation*}
	\langle ({\bf L}_{\varepsilon}^+)^{-1} {\bf u}, {\bf u} \rangle = - \frac{1}{2}\frac{d}{d\omega} \| {\bf u}(\omega) \|^2_{\ell^2} |_{\omega = 1},
	\end{equation*}
	where ${\bf u}(\omega) \in \ell^2(\mathbb{Z})$ is a solution of the difference equation augmented with the parameter $\omega \in \mathbb{R}$: 
	\begin{equation}
	\varepsilon \Delta_2 u_n(\omega) - \omega u_n(\omega) + |u_n(\omega) |^{p-1} u_n(\omega) - \gamma |u_n(\omega)|^{q-1} u_n(\omega) = 0.
	\label{second-order}
	\end{equation}
	By taking derivative of (\ref{second-order}) with respect to $\omega$ and evaluating it at $\omega = 1$, we obtain 
	\begin{equation*}
	{\bf L}_{\varepsilon}^+ {\bf u}'(\omega) |_{\omega = 1} = -{\bf u}, \quad \Rightarrow \quad {\bf u}'(\omega) |_{\omega = 1} = -  ({\bf L}_{\varepsilon}^+)^{-1} {\bf u},
	\end{equation*}
since ${\bf L}_{\varepsilon}^+$ is invertible. If 
$$
\frac{d}{d\omega} \| {\bf u}(\omega) \|^2_{\ell^2} |_{\omega = 1} > 0,
$$ 
then ${\bf u}$ is  local constrained minimizer of energy, and since $N_c = N_r^{\pm} = N_i^- = 0$ follow from (\ref{count-1})--(\ref{count-2}), the local constrained minimizer is spectrally stable. It is also nonlinearly (orbitally) stable (see Theorems 2.8 and 3.3 in \cite{GP25}).
\end{remark}

\begin{remark}
	Under the assumptions of Proposition \ref{theorem-count}, the zero 
	eigevalue of the generalized eigenvalue problem (\ref{Gener}) is simple. 
	The eigenvalues in $N_c$, $N_r^{\pm}$, $N_i^{\pm}$ could be multiple, in which case the positive (negative) values of $\langle ({\bf L}_{\varepsilon}^+)^{-1} {\bf w}, {\bf w} \rangle$ are defined on the invariant subspaces of $\ell^2(\mathbb{Z})$ related to multiple eigenvalues. 
	The case with $\langle ({\bf L}_{\varepsilon}^+)^{-1} {\bf u}, {\bf u} \rangle = 0$ is degenerate with the zero eigenvalue of (\ref{Gener}) being at least double. 
\end{remark}

As an application of Proposition \ref{theorem-count}, we consider spectral stability of the simplest ILMs with the shortest codes.

\begin{example}
	\label{ex-stable-1}
If $N = 1$, there are two such ILMs, $(A_+)$ and $(a_+)$, up to the sign reflection. 
	\begin{itemize}
		\item  For the code $(A_+)$, we have $K = 0$ and by (\ref{neg-3}), $\sigma_0 = 0$. Since $n({\bf L}^+_\varepsilon) = n({\bf L}^-_\varepsilon) = 0$, then $N_c = N_r^\pm = N_i^- = 0$ follow from (\ref{count-1})--(\ref{count-2}). Therefore  $(A_+)$ is spectrally stable (a local minimizer of $\Lambda$).
		\item  For the code $(a_+)$, we have $K =1$ and by (\ref{neg-3}), $\sigma_0 = 1$. Also $n({\bf L}^+_\varepsilon)=1$,  $n({\bf L}^-_\varepsilon) = 0$, then $N_c=N_r^\pm =N_i^-=0$ follow from (\ref{count-1})--(\ref{count-2}). This implies that  $(a_+)$ is spectrally stable (a local constrained minimizer of $H$ for fixed $Q$).
	\end{itemize}
\end{example}

\begin{remark}
	The stability of the two ILMs with codes $(A_+)$ and $(a_+)$ takes place for any $p,q\in \mathbb{N}$ and $2 \leq p < q$ and any $\gamma \in (0,\gamma_{{p,q}})$. See the blue color branches for codes $(A)$ and $(a)$ for small $\varepsilon > 0$ in Figure \ref{Fig:Bifurcations} below.
\end{remark}

\begin{example}
		\label{ex-stable-2}
	If $N = 2$, there are six irreducible ILMs (see Example \ref{ex-2}). We have definite results for only four ILMs.
\begin{itemize}
		\item For the code $(A_+A_+)$, we have $K = 0$ and by (\ref{neg-3}), $\sigma_0 = 0$. Then, $n({\bf L}^+_\varepsilon) = n({\bf L}^-_\varepsilon) = 0$ implies $N_c = N_r^{\pm} = N_i^- = 0$ by (\ref{count-1})--(\ref{count-2}). Hence, $(A_+A_+)$ is spectrally  stable (a local minimizer of $\Lambda$).
		
		\item For the code $(A_+A_-)$, we have $K = 0$ and by (\ref{neg-3}), $\sigma_0 = 0$. Then, $n({\bf L}^+_\varepsilon) = 0$ implies $N_c = N_r^- = N_i^- = 0$ from (\ref{count-1}) and $n({\bf L}^-_\varepsilon) = 1$ implies $N_r^+ = 1$ from (\ref{count-2}). Hence, $(A_+A_-)$ is spectrally unstable. 
		
		\item For the code $(a_+a_+)$, we have $K = 2$ and by (\ref{neg-3}), $\sigma_0 = 1$. Then, $n({\bf L}^+_\varepsilon) - \sigma_0 = 1$ 
		implies $N_r^- = 1$ from (\ref{count-1}) and $n({\bf L}^-_\varepsilon) = 0$ implies  $N_c = N_r^+ = N_i^- = 0$ from (\ref{count-2}).  Hence, $(a_+a_+)$ is spectrally unstable. 
		
		\item For the code $(a_+a_-)$, we have $K = 2$ and by (\ref{neg-3}), $\sigma_0 = 1$. Then, $n({\bf L}^+_\varepsilon) - \sigma_0 = n({\bf L}^-_\varepsilon) = 1$ imply $N_i^- = 1$ and $N_c = N_r^{\pm} = 0$ from (\ref{count-1})--(\ref{count-2}). The other possible cases 
		$N_c = 1$, $N_r^{\pm} = N_0^- = 0$ or $N_r^+ = N_r^- = 1$, $N_c = N_0^- = 0$ in (\ref{count-1})--(\ref{count-2}) are excluded since they require two nonzero eigenvalues in the truncated generalized eigenvalue problem (\ref{Gener-matrix}), whereas we only have one nonzero and one zero eigenvalue for $N = 2$. Hence, 
		$(a_+a_-)$ is spectrally stable (but it is not a constrained minimizer of $H$ for fixed $Q$).
\end{itemize} 
\end{example}
\begin{remark}
	Stability or instability of the ILMs listed above takes place  for any $p,q\in \mathbb{N}$ and $2 \leq p < q$ and any $\gamma \in (0,\gamma_{{p,q}})$. See the blue color branch for code $(AA)$ and the red color branch for code $(aa)$ for small $\varepsilon > 0$ in Figure \ref{Fig:Bifurcations} below.
\end{remark}

For the two remaining ILMs with $N = 2$ with codes $(a_+A_+)$ and $(a_+A_-)$, we have $K = 1$ and, hence, we need to compute the quantity in (\ref{neg-3}), that is, 
$$
(\tilde{\bf L}^+)^{-1} \tilde{\bf u} \cdot \tilde{\bf u} = \frac{a^2}{f'(a)} + \frac{A^2}{f'(A)}.
$$
The following lemma shows that the quantity is exactly zero for $(p,q) = (3,5)$, as a part of a more general statement. 

\begin{lemma}
	\label{lem-3-5}
	If $N = 2K$ and $(p,q) = (3,5)$, then $(\tilde{\bf L}^+)^{-1} \tilde{\bf u} \cdot \tilde{\bf u} = 0$ for any $\gamma \in (0,\gamma_{{3,5}})$.
\end{lemma}

\begin{proof}
We use (\ref{neg-3}) and write it explicitly  as
\begin{align*}
(\tilde{\bf L}^+)^{-1} \tilde{\bf u} \cdot \tilde{\bf u} &= \frac{Ka^2}{f'(a)} + \frac{KA^2}{f'(A)} = \frac{Ka^2}{-2a^2 + 4 \gamma a^4} + \frac{KA^2}{-2A^2 + 4 \gamma A^4} \\
&= -\frac{K(1-\gamma(a^2 + A^2))}{(1-2\gamma a^2)(1-2 \gamma A^2)}.
\end{align*}
Since $a^2 - \gamma a^4 = 1$, $A^2 - \gamma A^4 = 1$, 
we have $(a^2 - A^2) (1 - \gamma(a^2 + A^4)) = 0$ and since $a \neq A$, 
we must have $1 = \gamma (a^2 + A^2)$ which implies  $(\tilde{\bf L}^+)^{-1} \tilde{\bf u} \cdot \tilde{\bf u} = 0$ in this case.
\end{proof}

Consistently with Lemma \ref{lem-3-5}, we checked numerically $	(\tilde{\bf L}^+)^{-1} \tilde{\bf u} \cdot \tilde{\bf u} > 0$ for $(p,q) = (2,3)$ and $(p,q) = (3,4)$ and $(\tilde{\bf L}^+)^{-1} \tilde{\bf u} \cdot \tilde{\bf u} < 0$ 
for $(p,q) = (3,6)$. The result holds for any $\gamma \in (0,\gamma_{{p,q}})$. 
Given this crucial information, we complete the study of stability of the two 
remaining ILMs for these cases. 

\begin{example}
	\label{ex-stable-3}
	For $(p,q) = (2,3)$ and $(p,q) = (3,4)$, we have  $	(\tilde{\bf L}^+)^{-1} \tilde{\bf u} \cdot \tilde{\bf u} > 0$ which implies $\sigma_0 = 0$.  
\begin{itemize}
	\item For the code $(a_+A_+)$, we have $n({\bf L}^-_\varepsilon) = 0$ and $n({\bf L}^+_\varepsilon) = 1$ so that the counts (\ref{count-1})--(\ref{count-2}) with $\sigma_0 = 0$ imply 
	$N_c = N_r^+ = N_i^- =0$  and $N_r^-  = 1$. Hence, $(a_+A_+)$ is spectrally unstable. See the red color branch for code $(aA)$ in Figure \ref{Fig:Bifurcations} below.

	\item For the code $(a_+A_-)$, we have  $n({\bf L}^-_\varepsilon) = 1$ and $n({\bf L}^+_\varepsilon) = 1$ so that the counts (\ref{count-1})--(\ref{count-2}) with $\sigma_0 = 0$ imply $N_c = N_r^- = N_r^+ = 0$ and $N_i^- = 1$ by the same reasoning as in the code $(a_+,a_-)$. Hence, $(a_+A_-)$ is spectrally stable  (but it is not a constrained minimizer of $H$ for fixed $Q$).
\end{itemize}
For $(p,q)=(3,6)$, we have  $	(\tilde{\bf L}^+)^{-1} \tilde{\bf u} \cdot \tilde{\bf u} < 0$ which implies $\sigma_0 = 1$. 
\begin{itemize}
	\item For the code $(a_+A_+)$, we have $n({\bf L}^-_\varepsilon) = 0$ and $n({\bf L}^+_\varepsilon) = 1$ so that the counts (\ref{count-1})--(\ref{count-2}) with $\sigma_0 = 1$ imply 
	$N_c = N_r^+ = N_r^- = N_i^- = 0$. Hence, $(a_+A_+)$ is spectrally stable (a local constrained minimizer of $H$ for fixed $Q$).
	\item For the code $(a_+A_-)$, we have  $n({\bf L}^-_\varepsilon) = 1$ and $n({\bf L}^+_\varepsilon) = 1$ so that the counts (\ref{count-1})--(\ref{count-2}) with $\sigma_0 = 1$ imply $N_c = N_r^- =\textcolor{black}{N_i^-} = 0$ and $N_r^+ = 1$. Hence, $(a_+A_-)$ is spectrally unstable.
\end{itemize}
\end{example}

\begin{remark}
Since $(\tilde{\bf L}^+)^{-1} \tilde{\bf u} \cdot \tilde{\bf u} = 0$ for $(p,q) = (3,5)$ by Lemma \ref{lem-3-5}, the zero eigenvalue in the truncated generalized eigenvalue problem (\ref{Gener-matrix}) 
is at least double and no conclusive information is available for 
the full problem (\ref{Gener}) for small $\varepsilon > 0$ without computing higher orders of the perturbation theory.
\end{remark}

\section{Numerical results on continuation of ILMs for $\varepsilon>0$}
\label{Sect:Continuation}

The $\varepsilon$-depended branches of ILMs are characterized by their codes $\mathcal{A}$ as $\varepsilon \to 0$, see Proposition \ref{prop-existence}. When 
$\varepsilon$ grows, these branches may undergo bifurcations. Two types of bifurcations are common in the parameter continuations of ILMs and these bifurcations have been studied in \cite{AKA25}: 
\begin{itemize}
	\item the fold bifurcation (merging of two branches not related to each other with symmetries) and 
	\item the pitchfork bifurcation (two branches related to each other with either ${\bf R}$ or $-{\bf R}$ symmetry are connected to a branch of solutions that are invariant with respect to the same symmetry). 
\end{itemize}

It was found in \cite{AKA25} that the global picture of bifurcations is qualitatively the same for the cases $(p,q)=(2,3)$, $(p,q)=(3,4)$ and $(p,q)=(3,5)$, at least for  ILMs with codes of length $N\leq 3$. For different values of $\gamma$, the same branch may undergo different bifurcations. Figure~\ref{Fig:Two_bifurcations} illustrates this fact by the case of  ILMs with code $(A_+a_-)$ for $(p,q)=(3,4)$. If $\gamma=0.12$ the  branch with this code merges at $\varepsilon\approx 0.105$ with  the branch with code $(a_+A_+a_-)$ (the fold bifurcation). However, if $\gamma=0.22$, the pair of ILMs with codes $(A_+a_-)$ and $(a_+A_-)$ that are related to each other by $-{\bf R}$-symmetry is connected at $\varepsilon\approx0.099$ to the branches of $-{\bf R}$-symmetric ILMs with codes $(A_+A_-)$ and $(a_+a_-)$ (the pitchfork bifurcation).

\begin{figure}[htb!]
	\centerline{\includegraphics[width=0.85\textwidth]{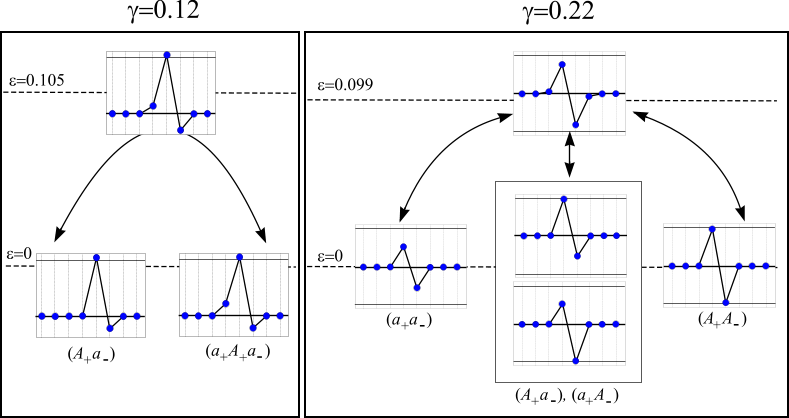}}
	\caption{Bifurcations of the branch of ILM with code $(A_+a_-)$, $(p,q)=(3,4)$. Left panel: $\gamma=0.12$. At $\varepsilon\approx 0,105$ this branch merges with  the branch with code $(a_+A_+a_-)$, the fold bifurcation.  Right panel:  $\gamma=0.22$.  At $\varepsilon\approx0.099$  the pair of ILMs with codes $(A_+a_-)$ and $(a_+A_-)$ is connected to the branches of $-R$-symmetric ILMs with codes $(A_+A_-)$ and $(a_+a_-)$, the pitchfork bifurcation. }
	\label{Fig:Two_bifurcations}
\end{figure}

\begin{figure}[htb!]
	\centerline{\includegraphics[width=0.9\textwidth]{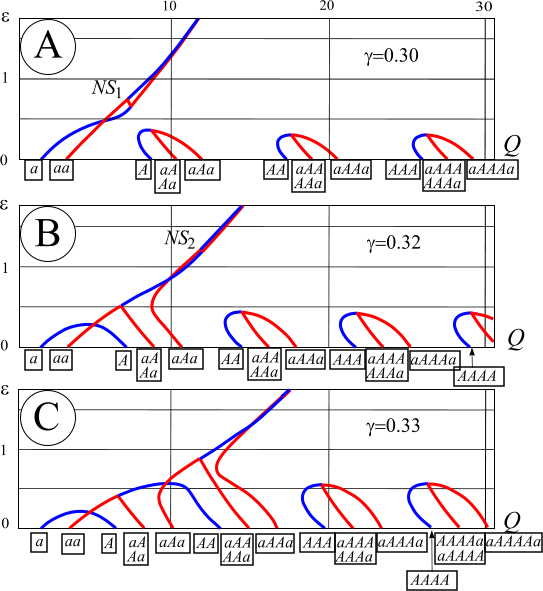}}
	\caption{Bifurcations of branches of ILMs for the case $(p,q) = (3,4)$: Panel A: $\gamma \approx 0.30$; panel B: $\gamma \approx0.32$; panel C:  $\gamma \approx0.33$. Values of $\varepsilon$ are shown versus 
		$Q({\bf u}) = \sum_{n \in \mathbb{Z}} u_n^2$.   Parts of the branches that correspond to stable (unstable) solutions are shown in blue (red).}
	\label{Fig:Bifurcations}
\end{figure}

Figure ~\ref{Fig:Bifurcations} shows the bifurcation diagram of  ILMs with short codes for the case $(p,q)=(3,4)$. Since onle positive solutions are considered,  the subscript ``$+$'' is omitted in the codes. One can see that the bifurcating counterparts depend on $\gamma$.  For instance, when $\gamma= 0.30$ the branch of ILM with code $(A_+)$ undergoes the fold bifurcation with the branch with code $(a_+A_+a_+)$ (see panel A), whereas  when $\gamma = 0.32$ and $\gamma = 0.33$ it undergoes the fold bifurcation with the branch with code $(a_+)$ (see panels B and C).  If $\gamma = 0.30$ the pair of branches of  ILMs, $(a_+A_+)$ and $(A_+a_+)$, related to each other by ${\bf R}$-symmetry,  undergoes the pitchfork bifurcation with the branch with code $(a_+A_+a_+)$ (see panel A). However, for $\gamma = 0.32$ and $\gamma = 0.33$ the same pair of branches undergoes the pitchfork bifurcation with the branch with code $(a_+a_+)$ (see panels B and C). 

For all the cases represented in Figure \ref{Fig:Bifurcations}, exactly two branches exist, that  extend to large values of $\varepsilon$ (called $\infty$-branches in \cite{AKA25}) . All other branches undergo bifurcations at finite values of $\varepsilon$. If $\gamma$ is far from the critical value $\gamma_{p,q}$ these two branches are the analogs of the Sievers-Takeno and the Page modes of the cubic DNLS equation (\ref{Eq:DNLS}), $(a_+)$ and $(a_+a_+)$, respectively. However, when $\gamma$ grows a cascade of switchings between the branches occurs that results in birth of new $\infty$-branches and death of old ones.  As a result, the length of the code of $\infty$-branches grows. They are $(a_+)$ and $(a_+a_+)$ at panel A, $(a_+a_+)$ and $(a_+A_+a_+)$ at panel B, and $(a_+A_+a_+)$ and $(a_+A_+A_+a_+)$ at panel C. We also show in panels (A) and (B) branches of asymmetric ILMs (denoted  $NS_1$ and $NS_2$ respectively), that disappear before the ACL. These branches of ILMs are typical for ``snaking'' phenomenon that was described for various cases of $(p,q)$,  \cite{TD10,CP11,Droplets02}.

Figure ~\ref{Fig:Bifurcations}  also provides insight into the stability or instability of the solution branches. One can see that the overall picture is rather complex, since even on the same branch the intervals of stability (blue) alternate with intervals of instability (red).  By this reason, we have restricted our study of stability of ILMs to the ACL.

\section{Numerical results on stable ILMs for small $\varepsilon > 0$}
\label{Sect:Total}

We employed numerical computations in order to identify all spectrally stable ILMs in the ACL. We computed eigenvalues of the truncated spectral problem 
(\ref{Gener-matrix}) with $0<\gamma<\gamma_{{p,q}}$. For better presentation, we denote $\gamma = \delta \gamma_{{p,q}}$ and consider $0<\delta<1$. Only irreducible codes of length $1 \leq N \leq 10$ have been studied. We consider the physically relevant cases $(p,q)=(2,3)$,  $(p,q)=(3,4)$ and $(p,q)=(3,5)$ but we also add data for $(p,q) = (3,6)$. Tables \ref{tab:2-3}, \ref{tab:3-4}, \ref{tab:3-5}, and \ref{tab:3-6} contain codes of spectrally stable ILMs for these four choices of $(p,q)$.

\begin{table}[ht]
\scriptsize
\centering
\begin{tabular}{|c|c|c|c|c|c|}
\toprule
 & \multicolumn{5}{c}{$\delta$,\quad\quad $\gamma=\delta\gamma_{{2,3}}$} \\
\cmidrule(lr){2-6}
$N$ & 0.2 &0.4 & 0.6 %& 0.80 
%& 0.88 
& 0.96 & 0.996 %& $\lim$  
\\ \hline
2 & 
\begin{tabular}{c}
$\mathcal{A}_A,$  $\mathcal{A}_a$,\\  
$\mathcal{A}_{1,1}^-$ \\ 
\end{tabular}
&
\begin{tabular}{c}
$\mathcal{A}_A,$  $\mathcal{A}_a$,\\  
$\mathcal{A}_{1,1}^-$ \\ 
\end{tabular}
&
\begin{tabular}{c}
$\mathcal{A}_A,$  $\mathcal{A}_a$,\\ 
$\mathcal{A}_{1,1}^-$ \\ 
\end{tabular}
%&
%\begin{tabular}{c}
%$\mathcal{A}_A$ , $\mathcal{A}_a$\\  
%$\mathcal{A}_{1,1}^-$ \\ 
%\end{tabular}
%0.88
%&
%\begin{tabular}{c}
%$\mathcal{A}_A,$  $\mathcal{A}_a$,\\  
%$\mathcal{A}_{1,1}^-$ \\ 
%\end{tabular}
&
\begin{tabular}{c}
$\mathcal{A}_A,$ $\mathcal{A}_a$,\\  
$\mathcal{A}_{1,1}^-$ \\ 
\end{tabular}
&
\begin{tabular}{c}
$\mathcal{A}_A,$  $\mathcal{A}_a$,\\  
$\mathcal{A}_{1,1}^-$ \\ 
\end{tabular}
%&
%\begin{tabular}{c}
%$\mathcal{A}_A,$  $\mathcal{A}_a$,\\  
%$\mathcal{A}_{1,1}^-,$ $\mathcal{A}_{1,1}^+$ \\ 
%\end{tabular}
\\
\hline
3 
& 
\begin{tabular}{c}
$\mathcal{A}_A,$ $\mathcal{A}_a$, \\ 
$\mathcal{A}_{2,1}^-$  $\mathcal{A}_{1,2}^-$\\ 
{\tiny $(a_+A_-a_+)$}\\ 
\end{tabular}
& 
\begin{tabular}{c}
$\mathcal{A}_A,$ $\mathcal{A}_a$, \\ 
$\mathcal{A}_{2,1}^-$  $\mathcal{A}_{1,2}^-$\\ 
{\tiny $(a_+A_-a_+)$}\\ 
\end{tabular}
&
\begin{tabular}{c}
$\mathcal{A}_A,$ $\mathcal{A}_a$, \\ 
$\mathcal{A}_{2,1}^-$  $\mathcal{A}_{1,2}^-$\\ 
{\tiny $(a_+A_-a_+)$}\\ 
\end{tabular}
%&
%\begin{tabular}{c}
%$\mathcal{A}_A$, $\mathcal{A}_a$ \\ 
%$\mathcal{A}_{2,1}^-$ \\ 
%\end{tabular}
%0.88
%&
%\begin{tabular}{c}
%$\mathcal{A}_A,$ $\mathcal{A}_a$, \\ 
%$\mathcal{A}_{2,1}^-$  $\mathcal{A}_{1,2}^-$\\ 
%{\tiny $(a_+A_-a_+)$}\\ 
%\end{tabular}
&
\begin{tabular}{c}
$\mathcal{A}_A,$ $\mathcal{A}_a$, \\ 
$\mathcal{A}_{2,1}^-$, $\mathcal{A}_{1,2}^+$\\ 
\end{tabular}
&
\begin{tabular}{c}
$\mathcal{A}_A,$ $\mathcal{A}_a$, \\ 
$\mathcal{A}_{2,1}^-$, $\mathcal{A}_{1,2}^+$\\ 
\end{tabular}
%&
%\begin{tabular}{c}
%$\mathcal{A}_A,$ $\mathcal{A}_a$, \\ 
%$\mathcal{A}_{2,1}^-$, $\mathcal{A}_{1,2}^+$ \\ 
%\end{tabular}
\\
\hline
4 
&
\begin{tabular}{c}
$\mathcal{A}_A,$ $\mathcal{A}_a$, \\ 
$\mathcal{A}_{1,3}^-$, $\mathcal{A}_{2,2}^-$,\\
{\tiny $(A_+a_-a_-A_+)$}\\ 
{\tiny $(a_+a_-A_+a_-)$}\\ 
{\tiny $(a_+A_-A_-a_+)$}\\ 
\end{tabular}
&
\begin{tabular}{c}
$\mathcal{A}_A,$ $\mathcal{A}_a$, \\ 
$\mathcal{A}_{1,3}^-$, $\mathcal{A}_{2,2}^-$,\\
{\tiny $(A_+a_-a_-A_+)$}\\ 
{\tiny $(a_+a_-A_+a_-)$}\\ 
{\tiny $(a_+A_-A_-a_+)$}\\ 
\end{tabular}
&
\begin{tabular}{c}
$\mathcal{A}_A,$ $\mathcal{A}_a$, \\ 
$\mathcal{A}_{1,3}^-$, $\mathcal{A}_{2,2}^-$,\\
{\tiny $(a_+a_-A_+a_-)$}\\ 
\end{tabular}
%&
%\begin{tabular}{c}
%$\mathcal{A}_A$, $\mathcal{A}_a$ \\ 
%$\mathcal{A}_{2,2}^-$\\
%\end{tabular}
%0.88
%&
%\begin{tabular}{c}
%$\mathcal{A}_A,$ $\mathcal{A}_a$ \\ 
%\end{tabular}
&
\begin{tabular}{c}
$\mathcal{A}_A,$ $\mathcal{A}_a$ \\ 
\end{tabular}
&
\begin{tabular}{c}
$\mathcal{A}_A,$ $\mathcal{A}_a$, \\ 
\end{tabular}
%&
%\begin{tabular}{c}
%$\mathcal{A}_A,$ $\mathcal{A}_a$, \\ 
%$\mathcal{A}_{2,2}^-$, $\mathcal{A}_{2,2}^+$,\\
%{\tiny + 2 more}\\
%\end{tabular}
\\ 
\hline
5
&
\begin{tabular}{c}
$\mathcal{A}_A,$ $\mathcal{A}_a$, \\ 
$\mathcal{A}_{2,3}^-$, $\mathcal{A}_{1,4}^-$, $\mathcal{A}_{4,1}^-$\\
{\tiny $(a_+a_-A_+A_+a_-)$},\\
{\tiny $(A_+A_+a_-a_-A_+)$}\\
{\tiny $(A_+a_-a_-A_+a_-)$}\\
{\tiny $(a_+a_-A_+a_-a_+)$}\\
{\tiny $(a_+a_-a_+A_-a_+)$}\\
%{\tiny +2 more}\\
%{\tiny +5 more}\\
\end{tabular}
&
\begin{tabular}{c}
$\mathcal{A}_A,$ $\mathcal{A}_a$, \\ 
$\mathcal{A}_{2,3}^-$, $\mathcal{A}_{1,4}^-$,\\
{\tiny $(a_+a_-A_+A_+a_-)$},\\
{\tiny $(A_+A_+a_-a_-A_+)$}\\
{\tiny $(A_+a_-a_-A_+a_-)$}\\
{\tiny $(a_+a_-A_+a_-a_+)$}\\
{\tiny $(a_+a_-a_+A_-a_+)$}\\
%{\tiny +2 more}\\
%{\tiny +5 more}\\
\end{tabular}
&
\begin{tabular}{c}
$\mathcal{A}_A,$ $\mathcal{A}_a$ \\  
$\mathcal{A}_{2,3}^-$, $\mathcal{A}_{1,4}^-$\\
{\tiny $(a_+a_-A_+a_-a_+)$},\\
{\tiny $(a_+a_-a_+A_-a_+)$},\\
%{\tiny +2 more}\\
\end{tabular}
%&
%\begin{tabular}{c}
%$\mathcal{A}_A$, $\mathcal{A}_a$ \\  
%\end{tabular}
%0.88
%&
%\begin{tabular}{c}
%$\mathcal{A}_A,$ $\mathcal{A}_a$ \\ 
%$\mathcal{A}_{3,2}^-$\\
%\end{tabular}
&
\begin{tabular}{c}
$\mathcal{A}_A,$ $\mathcal{A}_a$, \\ 
$\mathcal{A}_{3,2}^-$, $\mathcal{A}_{2,3}^-$,\\
$\mathcal{A}_{2,3}^+$\\ 
{\tiny $(A_+a_+a_-a_+A_-)$},\\
{\tiny $(A_+a_-A_+a_+a_-)$},\\
{\tiny $(a_+a_-A_-A_-a_+)$},\\
%{\tiny + 3 more}\\
\end{tabular}
&
\begin{tabular}{c}
$\mathcal{A}_A,$ $\mathcal{A}_a$, \\ 
$\mathcal{A}_{3,2}^-$\\
\end{tabular}
%&
%\begin{tabular}{c}
%$\mathcal{A}_A,$ $\mathcal{A}_a$, \\ 
%$\mathcal{A}_{3,2}^-,$ $\mathcal{A}_{2,3}^+$\\
%\end{tabular}
\\
\hline
6
&
\begin{tabular}{c}
$\mathcal{A}_A,$ $\mathcal{A}_a$, \\ 
$\mathcal{A}_{5,1}^-$, $\mathcal{A}_{4,2}^-$\\
 $\mathcal{A}_{2,4}^-$, $\mathcal{A}_{1,5}^-$\\
%{\tiny $(A_+A_+a_-a_-A_+A_+)$},\\
%{\tiny + 7 more}\\
{\tiny + 10 more}\\
\end{tabular}
&
\begin{tabular}{c}
$\mathcal{A}_A,$ $\mathcal{A}_a$, \\ 
$\mathcal{A}_{5,1}^-$, $\mathcal{A}_{4,2}^-$\\
 $\mathcal{A}_{2,4}^-$, $\mathcal{A}_{1,5}^-$\\
%{\tiny $(A_+A_+a_-a_-A_+A_+)$},\\
%{\tiny + 7 more}\\
{\tiny + 8 more}\\
\end{tabular}
&
\begin{tabular}{c}
$\mathcal{A}_A,$ $\mathcal{A}_a$ \\ 
$\mathcal{A}_{5,1}^-$, $\mathcal{A}_{4,2}^-$,\\
 $\mathcal{A}_{2,4}^-$\\
 {\tiny + 1 more}\\
%{\tiny $(A_+A_+a_-a_-A_+A_+)$},\\
\end{tabular}
%&
%\begin{tabular}{c}
%$\mathcal{A}_A$, $\mathcal{A}_a$ \\   
%\end{tabular}
%0.88
%&
%\begin{tabular}{c}
%$\mathcal{A}_A,$ $\mathcal{A}_a$, \\ 
%\end{tabular}
&
\begin{tabular}{c}
$\mathcal{A}_A,$ $\mathcal{A}_a$, \\ 
\end{tabular}
&
\begin{tabular}{c}
$\mathcal{A}_A,$ $\mathcal{A}_a$, \\ 
\end{tabular}
%&
%\begin{tabular}{c}
%$\mathcal{A}_A,$ $\mathcal{A}_a$, \\ 
%$\mathcal{A}_{3,3}^-$, $\mathcal{A}_{3,3}^+$\\
%{\tiny + 6 more}\\
%\end{tabular}
\\
\hline
7
&
\begin{tabular}{c}
$\mathcal{A}_A,$ $\mathcal{A}_a$, \\ 
$\mathcal{A}_{5,2}^-$, $\mathcal{A}_{4,3}^-$,\\
$\mathcal{A}_{2,5}^-$, $\mathcal{A}_{1,6}^-$\\
%{\tiny $(a_+a_-A_+A_+a_-a_+a_-)$},\\
%+ 8 more\\
{\tiny +15 more}\\
\end{tabular}
&
\begin{tabular}{c}
$\mathcal{A}_A,$ $\mathcal{A}_a$, \\ 
$\mathcal{A}_{4,3}^-$, $\mathcal{A}_{2,5}^-$, $\mathcal{A}_{1,6}^-$\\
%{\tiny $(a_+a_-A_+A_+a_-a_+a_-)$},\\
%+ 8 more\\
{\tiny +9 more}\\
\end{tabular}
&
\begin{tabular}{c}
$\mathcal{A}_A,$ $\mathcal{A}_a$ \\  
$\mathcal{A}_{2,5}^-$\\
\end{tabular}
%&
%\begin{tabular}{c}
%$\mathcal{A}_A$, $\mathcal{A}_a$ \\  
%\end{tabular}
%0.88
%&
%\begin{tabular}{c}
%$\mathcal{A}_A,$ $\mathcal{A}_a$ \\ 
%\end{tabular}
&
\begin{tabular}{c}
$\mathcal{A}_A,$ $\mathcal{A}_a$ \\ 
$\mathcal{A}_{4,3}^-$\\
\end{tabular}
&
\begin{tabular}{c}
$\mathcal{A}_A,$ $\mathcal{A}_a$, \\ 
$\mathcal{A}_{4,3}^-$\\
\end{tabular}
%&
%\begin{tabular}{c}
%$\mathcal{A}_A,$ $\mathcal{A}_a$, \\ 
%$\mathcal{A}_{4,3}^-$, $\mathcal{A}_{3,4}^+$\\
%\end{tabular}
\\
\hline
8
&
\begin{tabular}{c}
$\mathcal{A}_A,$ $\mathcal{A}_a$, \\
$\mathcal{A}_{4,4}^-$,  $\mathcal{A}_{2,6}^-$, $\mathcal{A}_{1,7}^-$\\ 
$\mathcal{A}_{7,1}^-$,  $\mathcal{A}_{5,3}^-$\\ 
%{\tiny $(A_+A_+A_+A_+A_+a_-a_-A_+)$}\\
%{\tiny $(a_+a_-A_+A_+A_+A_+a_-a_+)$}\\
%+ 13 more\\
{\tiny +23 more}\\
\end{tabular}
&
\begin{tabular}{c}
$\mathcal{A}_A,$ $\mathcal{A}_a$, \\
$\mathcal{A}_{4,4}^-$,  $\mathcal{A}_{2,6}^-$, $\mathcal{A}_{1,7}^-$\\ 
%{\tiny $(A_+A_+A_+A_+A_+a_-a_-A_+)$}\\
%{\tiny $(a_+a_-A_+A_+A_+A_+a_-a_+)$}\\
%+ 13 more\\
{\tiny +15 more}\\
\end{tabular}
&
\begin{tabular}{c}
$\mathcal{A}_A,$ $\mathcal{A}_a$, \\ 
$\mathcal{A}_{5,3}^-$\\
%{\tiny $(a_+a_-A_+A_+A_+A_+a_-a_+)$}\\
{\tiny + 1 more}\\
\end{tabular}
%&
%\begin{tabular}{c}
%$\mathcal{A}_A$, $\mathcal{A}_a$ \\   
%\end{tabular}
%0.88
%&
%\begin{tabular}{c}
%$\mathcal{A}_A,$ $\mathcal{A}_a$, \\ 
%$\mathcal{A}_{5,3}^-$\\
%\end{tabular}
&
\begin{tabular}{c}
$\mathcal{A}_A,$ $\mathcal{A}_a$\\  
\end{tabular}
&
\begin{tabular}{c}
$\mathcal{A}_A,$ $\mathcal{A}_a$, \\ 
\end{tabular}
%&
%\begin{tabular}{c}
%$\mathcal{A}_A,$ $\mathcal{A}_a$, \\ 
%$\mathcal{A}_{4,4}^-$, $\mathcal{A}_{4,4}^+$\\
%{\tiny+ 14 more}\\
%\end{tabular}
\\
\hline
9 
&
\begin{tabular}{c}
$\mathcal{A}_A,$ $\mathcal{A}_a$ ,\\ 
$\mathcal{A}_{2,7}^-$, $\mathcal{A}_{1,8}^-$,\\
$\mathcal{A}_{4,5}^-$, $\mathcal{A}_{5,4}^-$,\\
 $\mathcal{A}_{7,2}^-$ $\mathcal{A}_{8,1}^-$ \\  
%{\tiny $(A_+A_+A_+A_+A_+a_-a_-A_+A_+)$}\\
%{\tiny $(A_+A_+a_-a_-A_+A_+a_-a_-A_+)$}\\
%+ 13 more\\
{\tiny +29 more}\\
\end{tabular}
&
\begin{tabular}{c}
$\mathcal{A}_A,$ $\mathcal{A}_a$ ,\\ 
$\mathcal{A}_{2,7}^-$, $\mathcal{A}_{4,5}^-$,\\
 $\mathcal{A}_{7,2}^-$ $\mathcal{A}_{8,1}^-$ \\  
%{\tiny $(A_+A_+A_+A_+A_+a_-a_-A_+A_+)$}\\
%{\tiny $(A_+A_+a_-a_-A_+A_+a_-a_-A_+)$}\\
%+ 13 more\\
{\tiny +15 more}\\
\end{tabular}
&
\begin{tabular}{c}
$\mathcal{A}_A,$ $\mathcal{A}_a$ \\ 
$\mathcal{A}_{2,7}^-$, $\mathcal{A}_{4,5}^-$,\\
$\mathcal{A}_{8,1}^-$ \\  
\end{tabular}
%&
%\begin{tabular}{c}
%$\mathcal{A}_A$, $\mathcal{A}_a$ \\  
%\end{tabular}
%0.88
%&
%\begin{tabular}{c}
%$\mathcal{A}_A,$ $\mathcal{A}_a$ \\ 
%\end{tabular}
&
\begin{tabular}{c}
$\mathcal{A}_A,$ $\mathcal{A}_a$ \\  
\end{tabular}
&
\begin{tabular}{c}
$\mathcal{A}_A,$ $\mathcal{A}_a$, \\ 
$\mathcal{A}_{5,4}^-$\\
\end{tabular}
%&
%\begin{tabular}{c}
%$\mathcal{A}_A,$ $\mathcal{A}_a$, \\ 
%$\mathcal{A}_{5,4}^-,$ $\mathcal{A}_{4,5}^+$\\
%\end{tabular}
\\
\hline
10
&
\begin{tabular}{c}
$\mathcal{A}_A,$ $\mathcal{A}_a$, \\ 
$\mathcal{A}_{1,9}^-$, $\mathcal{A}_{5,5}^-$, $\mathcal{A}_{8,2}^-$\\
$\mathcal{A}_{7,3}^-$, $\mathcal{A}_{4,6}^-$, $\mathcal{A}_{2,8}^-$\\
%{\tiny $(A_+A_+A_+A_+A_+a_-a_-A_+A_+a_-)$}\\
%{\tiny $(A_+A_+a_-a_-A_+A_+a_-a_-A_+A_+)$}\\
%+ 17 more 
{\tiny +42 more}
\end{tabular}
&
\begin{tabular}{c}
$\mathcal{A}_A,$ $\mathcal{A}_a$, \\ 
$\mathcal{A}_{7,3}^-$, $\mathcal{A}_{4,6}^-$, $\mathcal{A}_{2,8}^-$\\
%{\tiny $(A_+A_+A_+A_+A_+a_-a_-A_+A_+a_-)$}\\
%{\tiny $(A_+A_+a_-a_-A_+A_+a_-a_-A_+A_+)$}\\
%+ 17 more 
{\tiny +19 more}
\end{tabular}
&
\begin{tabular}{c}
$\mathcal{A}_A,$ $\mathcal{A}_a$ \\  
$\mathcal{A}_{7,3}^-$, $\mathcal{A}_{2,8}^-$ \\
\end{tabular}
%&
%\begin{tabular}{c}
%$\mathcal{A}_A$, $\mathcal{A}_a$ \\ 
%\end{tabular}
%0.88
%&
%\begin{tabular}{c}
%$\mathcal{A}_A,$ $\mathcal{A}_a$ \\ 
%\end{tabular}
&
\begin{tabular}{c}
$\mathcal{A}_A,$ $\mathcal{A}_a$ \\ 
\end{tabular}
&
\begin{tabular}{c}
$\mathcal{A}_A,$ $\mathcal{A}_a$, \\   
\end{tabular}
%&
%\begin{tabular}{c}
%$\mathcal{A}_A,$ $\mathcal{A}_a$, \\   
%$\mathcal{A}_{5,5}^-,$ $\mathcal{A}_{5,5}^+$\\
%{\tiny + 6 more}\\
%\end{tabular}
\\
\hline
\end{tabular}
\vspace{0.25cm}
\caption{Codes of the spectrally stable ILMs for different values of parameter $\gamma$ in the quadratic-cubic case $(p,q) = (2,3)$.}
\label{tab:2-3}
\end{table}

\begin{table}[ht]
\scriptsize
\centering
\begin{tabular}{|c|c|c|c|c|c|}
\toprule
 & \multicolumn{5}{c}{$\delta$,\quad\quad $\gamma = \delta\gamma_{{3,4}}$} \\
\cmidrule(lr){2-6}
%\hline
$N$ &0.2 & 0.4 & 0.6 %& 0.20 
%& 0.88 
& 0.96 & 0.996 %& $\lim$  
\\ 
\hline
2 & 
\begin{tabular}{c}
$\mathcal{A}_A,$  $\mathcal{A}_a$,\\  
$\mathcal{A}_{1,1}^-$ \\ 
\end{tabular}
&
\begin{tabular}{c}
$\mathcal{A}_A,$  $\mathcal{A}_a$,\\  
$\mathcal{A}_{1,1}^-$ \\ 
\end{tabular}
&
\begin{tabular}{c}
$\mathcal{A}_A,$  $\mathcal{A}_a$,\\ 
$\mathcal{A}_{1,1}^-$ \\ 
\end{tabular}
%&
%\begin{tabular}{c}
%$\mathcal{A}_A$ , $\mathcal{A}_a$\\  
%$\mathcal{A}_{1,1}^-$ \\ 
%\end{tabular}
%&
%0.88
%\begin{tabular}{c}
%$\mathcal{A}_A,$  $\mathcal{A}_a$,\\  
%$\mathcal{A}_{1,1}^-$ \\ 
%\end{tabular}
&
\begin{tabular}{c}
$\mathcal{A}_A,$ $\mathcal{A}_a$,\\  
$\mathcal{A}_{1,1}^-$ \\ 
\end{tabular}
&
\begin{tabular}{c}
$\mathcal{A}_A,$  $\mathcal{A}_a$,\\  
$\mathcal{A}_{1,1}^-$ \\ 
\end{tabular}
% & \begin{tabular}{c} $\mathcal{A}_A,$  $\mathcal{A}_a$,\\  $\mathcal{A}_{1,1}^-,$ $\mathcal{A}_{1,1}^+$ \\ \end{tabular}
\\
\hline
3 
& 
\begin{tabular}{c}
$\mathcal{A}_A,$ $\mathcal{A}_a$, \\ 
$\mathcal{A}_{2,1}^-$ \\ 
{\tiny $(A_+a_-a_-)$}\\
\end{tabular}
& 
\begin{tabular}{c}
$\mathcal{A}_A,$ $\mathcal{A}_a$, \\ 
$\mathcal{A}_{2,1}^-$ \\ 
\end{tabular}
&
\begin{tabular}{c}
$\mathcal{A}_A,$ $\mathcal{A}_a$, \\ 
$\mathcal{A}_{2,1}^-$ \\ 
\end{tabular}
%&
%\begin{tabular}{c}
%$\mathcal{A}_A$, $\mathcal{A}_a$ \\ 
%$\mathcal{A}_{2,1}^-$ \\ 
%\end{tabular}
%0.88
%&
%\begin{tabular}{c}
%$\mathcal{A}_A,$ $\mathcal{A}_a$, \\ 
%$\mathcal{A}_{2,1}^-$ \\ 
%\end{tabular}
&
\begin{tabular}{c}
$\mathcal{A}_A,$ $\mathcal{A}_a$, \\ 
$\mathcal{A}_{2,1}^-$ \\ 
\end{tabular}
&
\begin{tabular}{c}
$\mathcal{A}_A,$ $\mathcal{A}_a$, \\ 
$\mathcal{A}_{2,1}^-$ \\ 
\end{tabular}
%&
%\begin{tabular}{c}
%$\mathcal{A}_A,$ $\mathcal{A}_a$, \\ 
%$\mathcal{A}_{2,1}^-$, $\mathcal{A}_{1,2}^+$ \\ 
%\end{tabular}
\\
\hline
4 
&
\begin{tabular}{c}
$\mathcal{A}_A,$ $\mathcal{A}_a$, \\ 
$\mathcal{A}_{3,1}^-$, $\mathcal{A}_{2,2}^-$,\\
{\tiny $(a_+A_-A_-a_+)$}\\ 
\end{tabular}
&
\begin{tabular}{c}
$\mathcal{A}_A,$ $\mathcal{A}_a$, \\ 
$\mathcal{A}_{3,1}^-$, $\mathcal{A}_{2,2}^-$,\\
{\tiny $(a_+A_-A_-a_+)$}\\ 
\end{tabular}
&
\begin{tabular}{c}
$\mathcal{A}_A,$ $\mathcal{A}_a$, \\ 
$\mathcal{A}_{3,1}^-$, $\mathcal{A}_{2,2}^-$,\\
{\tiny $(a_+A_-A_-a_+)$}\\ 
\end{tabular}
%&
%\begin{tabular}{c}
%$\mathcal{A}_A$, $\mathcal{A}_a$ \\ 
%$\mathcal{A}_{2,2}^-$\\
%\end{tabular}
%0.88
%&
%\begin{tabular}{c}
%$\mathcal{A}_A,$ $\mathcal{A}_a$ \\ 
%\end{tabular}
&
\begin{tabular}{c}
$\mathcal{A}_A,$ $\mathcal{A}_a$ \\ 
\end{tabular}
&
\begin{tabular}{c}
$\mathcal{A}_A,$ $\mathcal{A}_a$, \\ 
\end{tabular}
%&
%\begin{tabular}{c}
%$\mathcal{A}_A,$ $\mathcal{A}_a$, \\ 
%$\mathcal{A}_{2,2}^-$, $\mathcal{A}_{2,2}^+$,\\
%{\tiny + 2 more}\\
%\end{tabular}
\\ 
\hline
5
&
\begin{tabular}{c}
$\mathcal{A}_A,$ $\mathcal{A}_a$, \\ 
$\mathcal{A}_{3,2}^-$, $\mathcal{A}_{2,3}^-$, $\mathcal{A}_{4,1}^-$\\
{\tiny $(a_+a_-A_+A_+a_-)$},\\
{\tiny $(a_+A_-A_-A_-a_+)$}\\
\end{tabular}
&
\begin{tabular}{c}
$\mathcal{A}_A,$ $\mathcal{A}_a$, \\ 
$\mathcal{A}_{3,2}^-$, $\mathcal{A}_{2,3}^-$,\\
{\tiny $(a_+a_-A_+A_+a_-)$},\\
{\tiny $(a_+A_-A_-A_-a_+)$}\\
\end{tabular}
&
\begin{tabular}{c}
$\mathcal{A}_A,$ $\mathcal{A}_a$ \\   
\end{tabular}
%&
%\begin{tabular}{c}
%$\mathcal{A}_A$, $\mathcal{A}_a$ \\  
%\end{tabular}
%0.88
%&
%\begin{tabular}{c}
%$\mathcal{A}_A,$ $\mathcal{A}_a$ \\ 
%\end{tabular}
&
\begin{tabular}{c}
$\mathcal{A}_A,$ $\mathcal{A}_a$, \\ 
$\mathcal{A}_{3,2}^-$\\ 
\end{tabular}
&
\begin{tabular}{c}
$\mathcal{A}_A,$ $\mathcal{A}_a$, \\ 
$\mathcal{A}_{3,2}^-$\\
\end{tabular}
%&
%\begin{tabular}{c}
%$\mathcal{A}_A,$ $\mathcal{A}_a$, \\ 
%$\mathcal{A}_{3,2}^-,$ $\mathcal{A}_{2,3}^+$\\
%\end{tabular}
\\
\hline
6
&
\begin{tabular}{c}
$\mathcal{A}_A,$ $\mathcal{A}_a$, \\ 
$\mathcal{A}_{5,1}^-$, $\mathcal{A}_{3,3}^-$,$\mathcal{A}_{4,2}^-$\\
{\tiny $(a_+a_-A_+A_+A_+a_-)$}\\
{\tiny $(A_+A_+a_-a_-A_+A_+)$}\\
{\tiny $(a_+a_-a_-A_+A_+a_-)$}\\
{\tiny $(a_+A_-A_-A_-A_-a_+)$}\\
\end{tabular}
&
\begin{tabular}{c}
$\mathcal{A}_A,$ $\mathcal{A}_a$, \\ 
$\mathcal{A}_{3,3}^-$,\\
{\tiny $(a_+a_-A_+A_+A_+a_-)$}\\
\end{tabular}
&
\begin{tabular}{c}
$\mathcal{A}_A,$ $\mathcal{A}_a$ \\ 
\end{tabular}
%&
%\begin{tabular}{c}
%$\mathcal{A}_A$, $\mathcal{A}_a$ \\   
%\end{tabular}
%0.88
%&
%\begin{tabular}{c}
%$\mathcal{A}_A,$ $\mathcal{A}_a$, \\ 
%$\mathcal{A}_{4,2}^-$\\
%\end{tabular}
&
\begin{tabular}{c}
$\mathcal{A}_A,$ $\mathcal{A}_a$, \\ 
\end{tabular}
&
\begin{tabular}{c}
$\mathcal{A}_A,$ $\mathcal{A}_a$, \\ 
\end{tabular}
%&
%\begin{tabular}{c}
%$\mathcal{A}_A,$ $\mathcal{A}_a$, \\ 
%$\mathcal{A}_{3,3}^-$, $\mathcal{A}_{3,3}^+$\\
%{\tiny + 6 more}\\
%\end{tabular}
\\
\hline
7
&
\begin{tabular}{c}
$\mathcal{A}_A,$ $\mathcal{A}_a$, \\ 
$\mathcal{A}_{5,2}^-$,$\mathcal{A}_{4,3}^-$,$\mathcal{A}_{3,4}^-$,\\
{\tiny $(a_+a_-A_+A_+A_+a_-a_+)$},\\
{\tiny $(a_+a_-a_+A_-A_-A_-a_+)$}\\ 
{\tiny $(a_+a_-A_+A_+A_+A_+a_-)$}\\ 
{\tiny $(A_+A_+A_+a_-a_-A_+A_+)$}\\ 
\end{tabular}
&
\begin{tabular}{c}
$\mathcal{A}_A,$ $\mathcal{A}_a$, \\ 
$\mathcal{A}_{3,4}^-$,\\
{\tiny $(a_+a_-A_+A_+A_+a_-a_+)$},\\
{\tiny $(a_+a_-a_+A_-A_-A_-a_+)$}\\ 
\end{tabular}
&
\begin{tabular}{c}
$\mathcal{A}_A,$ $\mathcal{A}_a$ \\  
\end{tabular}
%&
%\begin{tabular}{c}
%$\mathcal{A}_A$, $\mathcal{A}_a$ \\  
%\end{tabular}
%0.88
%&
%\begin{tabular}{c}
%$\mathcal{A}_A,$ $\mathcal{A}_a$ \\ 
%\end{tabular}
&
\begin{tabular}{c}
$\mathcal{A}_A,$ $\mathcal{A}_a$ \\ 
\end{tabular}
&
\begin{tabular}{c}
$\mathcal{A}_A,$ $\mathcal{A}_a$, \\ 
$\mathcal{A}_{4,3}^-$\\
\end{tabular}
%&
%\begin{tabular}{c}
%$\mathcal{A}_A,$ $\mathcal{A}_a$, \\ 
%$\mathcal{A}_{4,3}^-$, $\mathcal{A}_{3,4}^+$\\
%\end{tabular}
\\
\hline
8
&
\begin{tabular}{c}
$\mathcal{A}_A,$ $\mathcal{A}_a$, \\ 
$\mathcal{A}_{5,3}^-$,$\mathcal{A}_{4,4}^-$,$\mathcal{A}_{3,5}^-$\\
{\tiny + 8 more}\\
\end{tabular}
&
\begin{tabular}{c}
$\mathcal{A}_A,$ $\mathcal{A}_a$, \\ 
\end{tabular}
&
\begin{tabular}{c}
$\mathcal{A}_A,$ $\mathcal{A}_a$, \\ 
$\mathcal{A}_{6,2}^-$\\
\end{tabular}
%&
%\begin{tabular}{c}
%$\mathcal{A}_A$, $\mathcal{A}_a$ \\   
%\end{tabular}
%0.88
%&
%\begin{tabular}{c}
%$\mathcal{A}_A,$ $\mathcal{A}_a$, \\ 
%\end{tabular}
&
\begin{tabular}{c}
$\mathcal{A}_A,$ $\mathcal{A}_a$\\  
\end{tabular}
&
\begin{tabular}{c}
$\mathcal{A}_A,$ $\mathcal{A}_a$, \\ 
\end{tabular}
%&
%\begin{tabular}{c}
%$\mathcal{A}_A,$ $\mathcal{A}_a$, \\ 
%$\mathcal{A}_{4,4}^-$, $\mathcal{A}_{4,4}^+$\\
%{\tiny+ 14 more}\\
%\end{tabular}
\\
\hline
9
&
\begin{tabular}{c}
$\mathcal{A}_A,$ $\mathcal{A}_a$ ,\\ 
$\mathcal{A}_{5,4}^-$,$\mathcal{A}_{4,5}^-$,$\mathcal{A}_{3,5}^-$\\
%{\tiny $(A_+A_+A_+A_+a_-a_-A_+A_+A_+)$}\\
{\tiny + 8 more}\\
\end{tabular}
&
\begin{tabular}{c}
$\mathcal{A}_A,$ $\mathcal{A}_a$ ,\\ 
%{\tiny $(A_+A_+A_+A_+a_-a_-A_+A_+A_+)$}\\
{\tiny + 1 more}\\
\end{tabular}
&
\begin{tabular}{c}
$\mathcal{A}_A,$ $\mathcal{A}_a$ \\ 
\end{tabular}
%&
%\begin{tabular}{c}
%$\mathcal{A}_A$, $\mathcal{A}_a$ \\  
%\end{tabular}
%0.88
%&
%\begin{tabular}{c}
%$\mathcal{A}_A,$ $\mathcal{A}_a$ \\ 
%\end{tabular}
&
\begin{tabular}{c}
$\mathcal{A}_A,$ $\mathcal{A}_a$ \\  
\end{tabular}
&
\begin{tabular}{c}
$\mathcal{A}_A,$ $\mathcal{A}_a$, \\ 
$\mathcal{A}_{5,4}^-$\\
\end{tabular}
%&
%\begin{tabular}{c}
%$\mathcal{A}_A,$ $\mathcal{A}_a$, \\ 
%$\mathcal{A}_{5,4}^-,$ $\mathcal{A}_{4,5}^+$\\
%\end{tabular}
\\
\hline
10
&
\begin{tabular}{c}
$\mathcal{A}_A,$ $\mathcal{A}_a$, \\ 
$\mathcal{A}_{5,5}^-$,$\mathcal{A}_{4,6}^-$\\
{\tiny + 12 more}\\
\end{tabular}
&
\begin{tabular}{c}
$\mathcal{A}_A,$ $\mathcal{A}_a$, \\ 
$\mathcal{A}_{8,2}^-$\\
\end{tabular}
&
\begin{tabular}{c}
$\mathcal{A}_A,$ $\mathcal{A}_a$ \\   
\end{tabular}
%&
%\begin{tabular}{c}
%$\mathcal{A}_A$, $\mathcal{A}_a$ \\ 
%\end{tabular}
%0.88
%&
%\begin{tabular}{c}
%$\mathcal{A}_A,$ $\mathcal{A}_a$ \\ 
%\end{tabular}
&
\begin{tabular}{c}
$\mathcal{A}_A,$ $\mathcal{A}_a$ \\ 
\end{tabular}
&
\begin{tabular}{c}
$\mathcal{A}_A,$ $\mathcal{A}_a$, \\   
\end{tabular}
%&
%\begin{tabular}{c}
%$\mathcal{A}_A,$ $\mathcal{A}_a$, \\   
%$\mathcal{A}_{5,5}^-,$ $\mathcal{A}_{5,5}^+$\\
%{\tiny + 6 more}\\
%\end{tabular}
\\
\hline
\end{tabular}
\vspace{0.25cm}
\caption{Codes of the spectrally stable ILMs for different values of parameter $\gamma$ in the cubic-quartic case $(p,q) = (3,4)$.}
\label{tab:3-4}
\end{table}
\normalsize

\begin{table}[ht]
\scriptsize
\centering
\begin{tabular}{|c|c|c|c|c|c|}
\toprule
 & \multicolumn{5}{c}{$\delta$,\quad\quad $\gamma = \delta\gamma_{3,5}$} \\
\cmidrule(lr){2-6}
%\hline
$N$ & 0.2 &0.4 & 0.6 %& 0.20 
%& 0.88 
& 0.96 & 0.996 %& $\lim$  
\\ 
\hline
2 &
\begin{tabular}{c}
$\mathcal{A}_A,$  $\mathcal{A}_a$,\\  
\textcolor{red}{$\mathcal{A}_{1,1}^-$}, \textcolor{red}{$\mathcal{A}_{1,1}^+$}\\
\end{tabular}
&
\begin{tabular}{c}
$\mathcal{A}_A,$  $\mathcal{A}_a$,\\  
\textcolor{red}{$\mathcal{A}_{1,1}^-$}, \textcolor{red}{$\mathcal{A}_{1,1}^+$}\\
\end{tabular}
&
\begin{tabular}{c}
$\mathcal{A}_A,$  $\mathcal{A}_a$,\\ 
\textcolor{red}{$\mathcal{A}_{1,1}^-$}, \textcolor{red}{$\mathcal{A}_{1,1}^+$}\\
\end{tabular}
%&
%\begin{tabular}{c}
%$\mathcal{A}_A$ , $\mathcal{A}_a$\\  
%$\mathcal{A}_{1,1}^-$ \\ 
%\end{tabular}
%0.88
%&
%\begin{tabular}{c}
%$\mathcal{A}_A,$  $\mathcal{A}_a$,\\  
%$\mathcal{A}_{1,1}^-$, $\mathcal{A}_{1,1}^+$\\
%\end{tabular}
&
\begin{tabular}{c}
$\mathcal{A}_A,$ $\mathcal{A}_a$,\\  
\textcolor{red}{$\mathcal{A}_{1,1}^-$}, \textcolor{red}{$\mathcal{A}_{1,1}^+$}\\
\end{tabular}
&
\begin{tabular}{c}
$\mathcal{A}_A,$  $\mathcal{A}_a$,\\  
\textcolor{red}{$\mathcal{A}_{1,1}^-$}, \textcolor{red}{$\mathcal{A}_{1,1}^+$}\\
\end{tabular}
% & \begin{tabular}{c} $\mathcal{A}_A,$  $\mathcal{A}_a$,\\  $\mathcal{A}_{1,1}^-,$ $\mathcal{A}_{1,1}^+$ \\ \end{tabular}
\\
\hline
3 &
\begin{tabular}{c}
$\mathcal{A}_A,$ $\mathcal{A}_a$, \\ 
$\mathcal{A}_{2,1}^-$ \\ 
\end{tabular}
& 
\begin{tabular}{c}
$\mathcal{A}_A,$ $\mathcal{A}_a$, \\ 
$\mathcal{A}_{2,1}^-$ \\ 
\end{tabular}
&
\begin{tabular}{c}
$\mathcal{A}_A,$ $\mathcal{A}_a$, \\ 
$\mathcal{A}_{2,1}^-$ \\ 
\end{tabular}
%&
%\begin{tabular}{c}
%$\mathcal{A}_A$, $\mathcal{A}_a$ \\ 
%$\mathcal{A}_{2,1}^-$ \\ 
%\end{tabular}
%0.88
%&
%\begin{tabular}{c}
%$\mathcal{A}_A,$ $\mathcal{A}_a$, \\ 
%$\mathcal{A}_{2,1}^-$ \\ 
%\end{tabular}
&
\begin{tabular}{c}
$\mathcal{A}_A,$ $\mathcal{A}_a$, \\ 
$\mathcal{A}_{2,1}^-$ \\ 
\end{tabular}
&
\begin{tabular}{c}
$\mathcal{A}_A,$ $\mathcal{A}_a$, \\ 
$\mathcal{A}_{2,1}^-$ \\ 
\end{tabular}
%&
%\begin{tabular}{c}
%$\mathcal{A}_A,$ $\mathcal{A}_a$, \\ 
%$\mathcal{A}_{2,1}^-$, $\mathcal{A}_{1,2}^+$ \\ 
%\end{tabular}
\\
\hline
4 &
\begin{tabular}{c}
$\mathcal{A}_A,$ $\mathcal{A}_a$, \\ 
$\mathcal{A}_{3,1}^-$, \textcolor{red}{$\mathcal{A}_{2,2}^-$},\textcolor{red}{$\mathcal{A}_{2,2}^+$}\\
\textcolor{red}{{\tiny $(a_+A_-A_-a_+)$}}\\ 
\textcolor{red}{{\tiny $(A_+A_+a_-a_-)$}}\\ 
\textcolor{red}{{\tiny $(A_+a_+a_-A_+)$}}\\ 
\textcolor{red}{{\tiny $(a_+A_+A_+a_-)$}}\\ 
\end{tabular}
&
\begin{tabular}{c}
$\mathcal{A}_A,$ $\mathcal{A}_a$, \\ 
$\mathcal{A}_{3,1}^-$, \textcolor{red}{$\mathcal{A}_{2,2}^-$}, \textcolor{red}{$\mathcal{A}_{2,2}^+$},\\
\textcolor{red}{{\tiny $(a_+A_-A_-a_+)$}}\\ 
\textcolor{red}{{\tiny $(A_+a_+a_-A_+)$}}\\ 
\textcolor{red}{{\tiny $(a_+A_+A_+a_-)$}}\\ 
\end{tabular}
&
\begin{tabular}{c}
$\mathcal{A}_A,$ $\mathcal{A}_a$, \\ 
\textcolor{red}{$\mathcal{A}_{2,2}^+$}, \textcolor{red}{$\mathcal{A}_{2,2}^-$},\\
\textcolor{red}{{\tiny $(a_+A_-A_-a_+)$}}\\ 
\textcolor{red}{{\tiny $(A_+a_+a_-A_+)$}}\\ 
\textcolor{red}{{\tiny $(a_+A_+A_+a_-)$}}\\ 
\end{tabular}
%&
%\begin{tabular}{c}
%$\mathcal{A}_A$, $\mathcal{A}_a$ \\ 
%$\mathcal{A}_{2,2}^-$\\
%\end{tabular}
%0.88
%&
%\begin{tabular}{c}
%$\mathcal{A}_A,$ $\mathcal{A}_a$ \\ 
%$\mathcal{A}_{2,2}^+$\\
%{\tiny $(A_+a_+a_-A_+)$}\\ 
%{\tiny $(a_+A_+A_+a_-)$}\\ 
%\end{tabular}
&
\begin{tabular}{c}
$\mathcal{A}_A,$ $\mathcal{A}_a$ \\ 
\textcolor{red}{$\mathcal{A}_{2,2}^+$}\\
\textcolor{red}{{\tiny $(A_+a_+a_-A_+)$}}\\ 
\textcolor{red}{{\tiny $(a_+A_+A_+a_-)$}}\\ 
\end{tabular}
&
\begin{tabular}{c}
$\mathcal{A}_A,$ $\mathcal{A}_a$, \\ 
\textcolor{red}{$\mathcal{A}_{2,2}^+$}\\
\textcolor{red}{{\tiny $(A_+a_+a_-A_+)$}}\\ 
\textcolor{red}{{\tiny $(a_+A_+A_+a_-)$}}\\ 
\end{tabular}
%&
%\begin{tabular}{c}
%$\mathcal{A}_A,$ $\mathcal{A}_a$, \\ 
%$\mathcal{A}_{2,2}^-$, $\mathcal{A}_{2,2}^+$,\\
%{\tiny + 2 more}\\
%\end{tabular}
\\ 
\hline
5 &
\begin{tabular}{c}
$\mathcal{A}_A,$ $\mathcal{A}_a$, \\ 
$\mathcal{A}_{4,1}^-$, $\mathcal{A}_{3,2}^-$,\\
{\tiny $(a_+A_-A_-A_-a_+)$}\\
\end{tabular}
&
\begin{tabular}{c}
$\mathcal{A}_A,$ $\mathcal{A}_a$, \\ 
\end{tabular}
&
\begin{tabular}{c}
$\mathcal{A}_A,$ $\mathcal{A}_a$ \\   
\end{tabular}
%&
%\begin{tabular}{c}
%$\mathcal{A}_A$, $\mathcal{A}_a$ \\  
%\end{tabular}
%0.88
%&
%\begin{tabular}{c}
%$\mathcal{A}_A,$ $\mathcal{A}_a$ \\ 
%\end{tabular}
&
\begin{tabular}{c}
$\mathcal{A}_A,$ $\mathcal{A}_a$, \\ 
$\mathcal{A}_{3,2}^-$\\ 
\end{tabular}
&
\begin{tabular}{c}
$\mathcal{A}_A,$ $\mathcal{A}_a$, \\ 
$\mathcal{A}_{3,2}^-$\\
\end{tabular}
%&
%\begin{tabular}{c}
%$\mathcal{A}_A,$ $\mathcal{A}_a$, \\ 
%$\mathcal{A}_{3,2}^-,$ $\mathcal{A}_{2,3}^+$\\
%\end{tabular}
\\
\hline
6 &
\begin{tabular}{c}
$\mathcal{A}_A,$ $\mathcal{A}_a$, \\ 
\textcolor{red}{$\mathcal{A}_{3,3}^-$},\\
\textcolor{red}{{\tiny $(A_+A_+A_+a_-a_-a_+)$}}\\
\textcolor{red}{{\tiny $(a_+a_-A_+A_+A_+a_-)$}}\\
\end{tabular}
&
\begin{tabular}{c}
$\mathcal{A}_A,$ $\mathcal{A}_a$, \\ 
\textcolor{red}{$\mathcal{A}_{3,3}^-$},\\
\textcolor{red}{{\tiny $(a_+a_-A_-A_-A_-a_+)$}}\\
\end{tabular}
&
\begin{tabular}{c}
$\mathcal{A}_A,$ $\mathcal{A}_a$ \\ 
\textcolor{red}{$\mathcal{A}_{3,3}^+$},\\
\textcolor{red}{{\tiny $(a_+a_-A_-A_-A_-a_+)$}}\\
\end{tabular}
%&
%\begin{tabular}{c}
%$\mathcal{A}_A$, $\mathcal{A}_a$ \\   
%\end{tabular}
%0.88
%&
%\begin{tabular}{c}
%$\mathcal{A}_A,$ $\mathcal{A}_a$, \\ 
%\end{tabular}
&
\begin{tabular}{c}
$\mathcal{A}_A,$ $\mathcal{A}_a$, \\ 
\textcolor{red}{$\mathcal{A}_{3,3}^+$},\\
\end{tabular}
&
\begin{tabular}{c}
$\mathcal{A}_A,$ $\mathcal{A}_a$, \\ 
\textcolor{red}{$\mathcal{A}_{3,3}^+$},\\
\end{tabular}
%&
%\begin{tabular}{c}
%$\mathcal{A}_A,$ $\mathcal{A}_a$, \\ 
%$\mathcal{A}_{3,3}^-$, $\mathcal{A}_{3,3}^+$\\
%{\tiny + 6 more}\\
%\end{tabular}
\\
\hline
7 &
\begin{tabular}{c}
$\mathcal{A}_A,$ $\mathcal{A}_a$, \\ 
\end{tabular}
&
\begin{tabular}{c}
$\mathcal{A}_A,$ $\mathcal{A}_a$, \\ 
\end{tabular}
&
\begin{tabular}{c}
$\mathcal{A}_A,$ $\mathcal{A}_a$ \\  
$\mathcal{A}_{5,2}^-$\\
\end{tabular}
%&
%\begin{tabular}{c}
%$\mathcal{A}_A$, $\mathcal{A}_a$ \\  
%\end{tabular}
%0.88
%&
%\begin{tabular}{c}
%$\mathcal{A}_A,$ $\mathcal{A}_a$ \\ 
%\end{tabular}
&
\begin{tabular}{c}
$\mathcal{A}_A,$ $\mathcal{A}_a$ \\ 
$\mathcal{A}_{4,3}^-$\\
\end{tabular}
&
\begin{tabular}{c}
$\mathcal{A}_A,$ $\mathcal{A}_a$, \\ 
$\mathcal{A}_{4,3}^-$\\
\end{tabular}
%&
%\begin{tabular}{c}
%$\mathcal{A}_A,$ $\mathcal{A}_a$, \\ 
%$\mathcal{A}_{4,3}^-$, $\mathcal{A}_{3,4}^+$\\
%\end{tabular}
\\
\hline
8 &
\begin{tabular}{c}
$\mathcal{A}_A,$ $\mathcal{A}_a$, \\ 
\textcolor{red}{$\mathcal{A}_{4,4}^-$}\\
\end{tabular}
&
\begin{tabular}{c}
$\mathcal{A}_A,$ $\mathcal{A}_a$, \\ 
$\mathcal{A}_{6,2}^-$\\
\end{tabular}
&
\begin{tabular}{c}
$\mathcal{A}_A,$ $\mathcal{A}_a$, \\ 
\end{tabular}
%&
%\begin{tabular}{c}
%$\mathcal{A}_A$, $\mathcal{A}_a$ \\   
%\end{tabular}
%0.88
%&
%\begin{tabular}{c}
%$\mathcal{A}_A,$ $\mathcal{A}_a$, \\ 
%$\mathcal{A}_{5,3}^-$\\
%\end{tabular}
&
\begin{tabular}{c}
$\mathcal{A}_A,$ $\mathcal{A}_a$\\  
\end{tabular}
&
\begin{tabular}{c}
$\mathcal{A}_A,$ $\mathcal{A}_a$, \\ 
\textcolor{red}{$\mathcal{A}_{4,4}^+$}\\
\end{tabular}
%&
%\begin{tabular}{c}
%$\mathcal{A}_A,$ $\mathcal{A}_a$, \\ 
%$\mathcal{A}_{4,4}^-$, $\mathcal{A}_{4,4}^+$\\
%{\tiny+ 14 more}\\
%\end{tabular}
\\
\hline
9 &
\begin{tabular}{c}
$\mathcal{A}_A,$ $\mathcal{A}_a$ \\ 
\end{tabular}
&
\begin{tabular}{c}
$\mathcal{A}_A,$ $\mathcal{A}_a$ ,\\ 
\end{tabular}
&
\begin{tabular}{c}
$\mathcal{A}_A,$ $\mathcal{A}_a$ \\ 
\end{tabular}
%&
%\begin{tabular}{c}
%$\mathcal{A}_A$, $\mathcal{A}_a$ \\  
%\end{tabular}
%0.88
%&
%\begin{tabular}{c}
%$\mathcal{A}_A,$ $\mathcal{A}_a$ \\ 
%\end{tabular}
&
\begin{tabular}{c}
$\mathcal{A}_A,$ $\mathcal{A}_a$ \\  
\end{tabular}
&
\begin{tabular}{c}
$\mathcal{A}_A,$ $\mathcal{A}_a$, \\ 
$\mathcal{A}_{5,4}^-$\\
\end{tabular}
%&
%\begin{tabular}{c}
%$\mathcal{A}_A,$ $\mathcal{A}_a$, \\ 
%$\mathcal{A}_{5,4}^-,$ $\mathcal{A}_{4,5}^+$\\
%\end{tabular}
\\
\hline
10 &\begin{tabular}{c}
$\mathcal{A}_A,$ $\mathcal{A}_a$, \\ 
\end{tabular}
&
\begin{tabular}{c}
$\mathcal{A}_A,$ $\mathcal{A}_a$, \\ 
\end{tabular}
&
\begin{tabular}{c}
$\mathcal{A}_A,$ $\mathcal{A}_a$ \\   
\end{tabular}
%&
%\begin{tabular}{c}
%$\mathcal{A}_A$, $\mathcal{A}_a$ \\ 
%\end{tabular}
%0.88
%&
%\begin{tabular}{c}
%$\mathcal{A}_A,$ $\mathcal{A}_a$ \\ 
%\end{tabular}
&
\begin{tabular}{c}
$\mathcal{A}_A,$ $\mathcal{A}_a$ \\ 
\end{tabular}
&
\begin{tabular}{c}
$\mathcal{A}_A,$ $\mathcal{A}_a$, \\  
\textcolor{red}{$\mathcal{A}_{5,5}^+$}\\ 
\end{tabular}
%&
%\begin{tabular}{c}
%$\mathcal{A}_A,$ $\mathcal{A}_a$, \\   
%$\mathcal{A}_{5,5}^-,$ $\mathcal{A}_{5,5}^+$\\
%{\tiny + 6 more}\\
%\end{tabular}
\\
\hline
\end{tabular}
\vspace{0.25cm}
\caption{Codes of the spectrally stable ILMs for different values of parameter $\gamma$ in the cubic-quintic case $(p,q) = (3,5)$. }
\label{tab:3-5}
\end{table}
\normalsize

\begin{table}[ht]
\scriptsize
\centering
\begin{tabular}{|c|c|c|c|c|c|}
\toprule
 & \multicolumn{5}{c}{$\delta$,\quad\quad $\gamma = \delta\gamma_{3,6}$} \\
\cmidrule(lr){2-6}
%\hline
$N$ & 0.2 &0.4 & 0.6 %& 0.20 
%& 0.88 
& 0.96 & 0.996 %& $\lim$  
\\ 
\hline
2 &
\begin{tabular}{c}
$\mathcal{A}_A,$  $\mathcal{A}_a$,\\  
 $\mathcal{A}_{1,1}^+$\\
\end{tabular}
&
\begin{tabular}{c}
$\mathcal{A}_A,$  $\mathcal{A}_a$,\\  
$\mathcal{A}_{1,1}^+$\\
\end{tabular}
&
\begin{tabular}{c}
$\mathcal{A}_A,$  $\mathcal{A}_a$,\\ 
$\mathcal{A}_{1,1}^+$\\
\end{tabular}
%&
%\begin{tabular}{c}
%$\mathcal{A}_A$ , $\mathcal{A}_a$\\  
%$\mathcal{A}_{1,1}^-$ \\ 
%\end{tabular}
%0.88
%&
%\begin{tabular}{c}
%$\mathcal{A}_A,$  $\mathcal{A}_a$,\\  
%$\mathcal{A}_{1,1}^+$\\
%\end{tabular}
&
\begin{tabular}{c}
$\mathcal{A}_A,$ $\mathcal{A}_a$,\\  
$\mathcal{A}_{1,1}^+$\\
\end{tabular}
&
\begin{tabular}{c}
$\mathcal{A}_A,$  $\mathcal{A}_a$,\\  
$\mathcal{A}_{1,1}^+$\\
\end{tabular}
% & \begin{tabular}{c} $\mathcal{A}_A,$  $\mathcal{A}_a$,\\  $\mathcal{A}_{1,1}^-,$ $\mathcal{A}_{1,1}^+$ \\ \end{tabular}
\\
\hline
3 &
\begin{tabular}{c}
$\mathcal{A}_A,$ $\mathcal{A}_a$, \\ 
$\mathcal{A}_{2,1}^-$ \\ 
\end{tabular}
& 
\begin{tabular}{c}
$\mathcal{A}_A,$ $\mathcal{A}_a$, \\ 
$\mathcal{A}_{2,1}^-$ \\ 
\end{tabular}
&
\begin{tabular}{c}
$\mathcal{A}_A,$ $\mathcal{A}_a$, \\ 
$\mathcal{A}_{2,1}^-$ \\ 
\end{tabular}
%&
%\begin{tabular}{c}
%$\mathcal{A}_A$, $\mathcal{A}_a$ \\ 
%$\mathcal{A}_{2,1}^-$ \\ 
%\end{tabular}
%0.88
%&
%\begin{tabular}{c}
%$\mathcal{A}_A,$ $\mathcal{A}_a$, \\ 
%$\mathcal{A}_{2,1}^-$ \\ 
%\end{tabular}
&
\begin{tabular}{c}
$\mathcal{A}_A,$ $\mathcal{A}_a$, \\ 
$\mathcal{A}_{2,1}^-$ \\ 
\end{tabular}
&
\begin{tabular}{c}
$\mathcal{A}_A,$ $\mathcal{A}_a$, \\ 
$\mathcal{A}_{2,1}^-$ \\ 
\end{tabular}
%&
%\begin{tabular}{c}
%$\mathcal{A}_A,$ $\mathcal{A}_a$, \\ 
%$\mathcal{A}_{2,1}^-$, $\mathcal{A}_{1,2}^+$ \\ 
%\end{tabular}
\\
\hline
4 &
\begin{tabular}{c}
$\mathcal{A}_A,$ $\mathcal{A}_a$, \\ 
$\mathcal{A}_{3,1}^-$\\ 
\end{tabular}
&
\begin{tabular}{c}
$\mathcal{A}_A,$ $\mathcal{A}_a$, \\ 
$\mathcal{A}_{3,1}^-$\\ 
\end{tabular}
&
\begin{tabular}{c}
$\mathcal{A}_A,$ $\mathcal{A}_a$, \\ 
$\mathcal{A}_{2,2}^+$\\
\end{tabular}
%&
%\begin{tabular}{c}
%$\mathcal{A}_A$, $\mathcal{A}_a$ \\ 
%$\mathcal{A}_{2,2}^-$\\
%\end{tabular}
%0.88
%&
%\begin{tabular}{c}
%$\mathcal{A}_A,$ $\mathcal{A}_a$ \\ 
%$\mathcal{A}_{2,2}^+$\\
%\end{tabular}
&
\begin{tabular}{c}
$\mathcal{A}_A,$ $\mathcal{A}_a$ \\ 
$\mathcal{A}_{2,2}^+$\\

\end{tabular}
&
\begin{tabular}{c}
$\mathcal{A}_A,$ $\mathcal{A}_a$, \\ 
$\mathcal{A}_{2,2}^+$\\
\end{tabular}
%&
%\begin{tabular}{c}
%$\mathcal{A}_A,$ $\mathcal{A}_a$, \\ 
%$\mathcal{A}_{2,2}^-$, $\mathcal{A}_{2,2}^+$,\\
%{\tiny + 2 more}\\
%\end{tabular}
\\ 
\hline
5 &
\begin{tabular}{c}
$\mathcal{A}_A,$ $\mathcal{A}_a$, \\ 
$\mathcal{A}_{3,2}^-$,\\
\end{tabular}
&
\begin{tabular}{c}
$\mathcal{A}_A,$ $\mathcal{A}_a$, \\ 
\end{tabular}
&
\begin{tabular}{c}
$\mathcal{A}_A,$ $\mathcal{A}_a$ \\   
\end{tabular}
%&
%\begin{tabular}{c}
%$\mathcal{A}_A$, $\mathcal{A}_a$ \\  
%\end{tabular}
%0.88
%&
%\begin{tabular}{c}
%$\mathcal{A}_A,$ $\mathcal{A}_a$ \\ 
%$\mathcal{A}_{3,2}^-$,\\
%\end{tabular}
&
\begin{tabular}{c}
$\mathcal{A}_A,$ $\mathcal{A}_a$, \\ 
$\mathcal{A}_{3,2}^-$\\ 
\end{tabular}
&
\begin{tabular}{c}
$\mathcal{A}_A,$ $\mathcal{A}_a$, \\ 
$\mathcal{A}_{3,2}^-$\\
\end{tabular}
%&
%\begin{tabular}{c}
%$\mathcal{A}_A,$ $\mathcal{A}_a$, \\ 
%$\mathcal{A}_{3,2}^-,$ $\mathcal{A}_{2,3}^+$\\
%\end{tabular}
\\
\hline
6 &
\begin{tabular}{c}
$\mathcal{A}_A,$ $\mathcal{A}_a$, \\ 
\end{tabular}
&
\begin{tabular}{c}
$\mathcal{A}_A,$ $\mathcal{A}_a$, \\ 
\end{tabular}
&
\begin{tabular}{c}
$\mathcal{A}_A,$ $\mathcal{A}_a$ \\ 
\end{tabular}
%&
%\begin{tabular}{c}
%$\mathcal{A}_A$, $\mathcal{A}_a$ \\   
%\end{tabular}
%0.88
%&
%\begin{tabular}{c}
%$\mathcal{A}_A,$ $\mathcal{A}_a$, \\ 
%\end{tabular}
&
\begin{tabular}{c}
$\mathcal{A}_A,$ $\mathcal{A}_a$, \\ 
$\mathcal{A}_{3,3}^+$,\\
\end{tabular}
&
\begin{tabular}{c}
$\mathcal{A}_A,$ $\mathcal{A}_a$, \\ 
$\mathcal{A}_{3,3}^+$,\\
\end{tabular}
%&
%\begin{tabular}{c}
%$\mathcal{A}_A,$ $\mathcal{A}_a$, \\ 
%$\mathcal{A}_{3,3}^-$, $\mathcal{A}_{3,3}^+$\\
%{\tiny + 6 more}\\
%\end{tabular}
\\
\hline
7 &
\begin{tabular}{c}
$\mathcal{A}_A,$ $\mathcal{A}_a$, \\ 
\end{tabular}
&
\begin{tabular}{c}
$\mathcal{A}_A,$ $\mathcal{A}_a$, \\ 
\end{tabular}
&
\begin{tabular}{c}
$\mathcal{A}_A,$ $\mathcal{A}_a$ \\  
\end{tabular}
%&
%\begin{tabular}{c}
%$\mathcal{A}_A$, $\mathcal{A}_a$ \\  
%\end{tabular}
%0.88
%&
%\begin{tabular}{c}
%$\mathcal{A}_A,$ $\mathcal{A}_a$ \\ 
%\end{tabular}
&
\begin{tabular}{c}
$\mathcal{A}_A,$ $\mathcal{A}_a$ \\ 
$\mathcal{A}_{4,3}^-$\\
\end{tabular}
&
\begin{tabular}{c}
$\mathcal{A}_A,$ $\mathcal{A}_a$, \\ 
$\mathcal{A}_{4,3}^-$\\
\end{tabular}
%&
%\begin{tabular}{c}
%$\mathcal{A}_A,$ $\mathcal{A}_a$, \\ 
%$\mathcal{A}_{4,3}^-$, $\mathcal{A}_{3,4}^+$\\
%\end{tabular}
\\
\hline
8 &
\begin{tabular}{c}
$\mathcal{A}_A,$ $\mathcal{A}_a$, \\ 
\end{tabular}
&
\begin{tabular}{c}
$\mathcal{A}_A,$ $\mathcal{A}_a$, \\ 
\end{tabular}
&
\begin{tabular}{c}
$\mathcal{A}_A,$ $\mathcal{A}_a$, \\ 
\end{tabular}
%&
%\begin{tabular}{c}
%$\mathcal{A}_A$, $\mathcal{A}_a$ \\   
%\end{tabular}
%0.88
%&
%\begin{tabular}{c}
%$\mathcal{A}_A,$ $\mathcal{A}_a$, \\ 
%$\mathcal{A}_{5,3}^-$\\
%\end{tabular}
&
\begin{tabular}{c}
$\mathcal{A}_A,$ $\mathcal{A}_a$\\  
$\mathcal{A}_{4,4}^+$\\
\end{tabular}
&
\begin{tabular}{c}
$\mathcal{A}_A,$ $\mathcal{A}_a$, \\ 
$\mathcal{A}_{4,4}^+$\\
\end{tabular}
%&
%\begin{tabular}{c}
%$\mathcal{A}_A,$ $\mathcal{A}_a$, \\ 
%$\mathcal{A}_{4,4}^-$, $\mathcal{A}_{4,4}^+$\\
%{\tiny+ 14 more}\\
%\end{tabular}
\\
\hline
9 &
\begin{tabular}{c}
$\mathcal{A}_A,$ $\mathcal{A}_a$ \\ 
\end{tabular}
&
\begin{tabular}{c}
$\mathcal{A}_A,$ $\mathcal{A}_a$ ,\\ 
\end{tabular}
&
\begin{tabular}{c}
$\mathcal{A}_A,$ $\mathcal{A}_a$ \\ 
\end{tabular}
%&
%\begin{tabular}{c}
%$\mathcal{A}_A$, $\mathcal{A}_a$ \\  
%\end{tabular}
%0.88
%&
%\begin{tabular}{c}
%$\mathcal{A}_A,$ $\mathcal{A}_a$ \\ 
%\end{tabular}
&
\begin{tabular}{c}
$\mathcal{A}_A,$ $\mathcal{A}_a$ \\  
\end{tabular}
&
\begin{tabular}{c}
$\mathcal{A}_A,$ $\mathcal{A}_a$, \\ 
$\mathcal{A}_{5,4}^-$\\
\end{tabular}
%&
%\begin{tabular}{c}
%$\mathcal{A}_A,$ $\mathcal{A}_a$, \\ 
%$\mathcal{A}_{5,4}^-,$ $\mathcal{A}_{4,5}^+$\\
%\end{tabular}
\\
\hline
10 &\begin{tabular}{c}
$\mathcal{A}_A,$ $\mathcal{A}_a$, \\ 
\end{tabular}
&
\begin{tabular}{c}
$\mathcal{A}_A,$ $\mathcal{A}_a$, \\ 
\end{tabular}
&
\begin{tabular}{c}
$\mathcal{A}_A,$ $\mathcal{A}_a$ \\   
\end{tabular}
%&
%\begin{tabular}{c}
%$\mathcal{A}_A$, $\mathcal{A}_a$ \\ 
%\end{tabular}
%0.88
%&
%\begin{tabular}{c}
%$\mathcal{A}_A,$ $\mathcal{A}_a$ \\ 
%\end{tabular}
&
\begin{tabular}{c}
$\mathcal{A}_A,$ $\mathcal{A}_a$ \\ 
\end{tabular}
&
\begin{tabular}{c}
$\mathcal{A}_A,$ $\mathcal{A}_a$, \\  
$\mathcal{A}_{5,5}^+$\\ 
\end{tabular}
%&
%\begin{tabular}{c}
%$\mathcal{A}_A,$ $\mathcal{A}_a$, \\   
%$\mathcal{A}_{5,5}^-,$ $\mathcal{A}_{5,5}^+$\\
%{\tiny + 6 more}\\
%\end{tabular}
\\
\hline
\end{tabular}
\vspace{0.25cm}
\caption{Codes of the spectrally stable ILMs for different values of parameter $\gamma$ in the cubic-sextic case $(p,q) = (3,6)$.}
\label{tab:3-6}
\end{table}
\normalsize

If $N = 1$, the only two stable ILMs are given by the codes $(A_+)$ and $(a_+)$ as in Example \ref{ex-stable-1}, where the code $(A_+)$ corresponds to the energy minimizer and the code $(a_+)$ corresponds to the constrained energy minimizer, both are spectrally and orbitally stable for every $(p,q)$ with $2\leq p<q$.  The codes $(A_-)$ and $(a_-)$ related to $(A_+)$ and $(a_+)$ by the sign-reversing symmetry are also stable.

If $N = 2$, the ILMs with codes $(A_+A_+)$ and $(a_+a_-)$ are stable for any $(p,q)$ with $2\leq p<q$ as in Example \ref{ex-stable-2}. In addition, 
it follows from Example \ref{ex-stable-3} that $(a_+A_-)$ is stable for  $(p,q)=(2,3)$ and  $(p,q)=(3,4)$, whereas $(a_+A_+)$ is stable for $(p,q) = (3,6)$. Stability of both codes is inconclusive for $(p,q) = (3,5)$ due to a multiple (double) zero eigenvalue in the truncated spectral problem (\ref{Gener-matrix}). All other codes of length $N = 2$ are unstable. 

For larger values, for $N \geq 3$, we summarize the following observations. 
\begin{itemize}
\item For $(p,q)=(2,3)$ and $(p,q)=(3,4)$ the number of spectrally stable codes  quickly grows when $\gamma \to 0$ (see Tables \ref{tab:2-3} and \ref{tab:3-4}). In comparison, there are very few spectrally stable codes for values of $\gamma$ near $\gamma_{p,q}$.

\item If $p=3$ and $q$ grows, $q=4,5,6$,   (see Tables \ref{tab:3-4}, \ref{tab:3-5}, and \ref{tab:3-6}), the number of stable codes decreases. For any value $q=4,5,6$ and for any $\gamma\in(0;\gamma_{p,q})$  the codes $\mathcal{A}_A$ and $\mathcal{A}_a$ remain stable. For $(p,q)=(3,5)$ and $(p,q)=(3,6)$ and $N=9, 10$ the codes $\mathcal{A}_A$ and $\mathcal{A}_a$ are the only stable ILMs for all $\gamma\in(0;\gamma_{p,q})$, except a small vicinity of $\gamma_{p,q}$.

\item If $N=2k+1$ the codes  $\mathcal{A}_{k+1,k}^-$ are  stable in small vicinity of $\gamma_{p,q}$.

\item If $N=2k$ the codes $\mathcal{A}_{k,k}^+$ are stable in small vicinity of $\gamma_{p,q}$ if $q$ exceeds some threshold: $q > 6$ for $p = 2$, 
$q > 5$ for $p = 3$, and any $q > p$ for $p \geq 4$. This fact is illustrated by Figure~\ref{Fig:5_5+}, where eigenvalues of the truncated spectral problem (\ref{Gener-matrix}) are plotted versus $\gamma$ in $(0,\gamma_{p,q})$ for $\mathcal{A}_{5,5}^+$ in three cases: below the threshold $(p,q)=(3,4)$, on the threshold  $(p,q)=(3,5)$, and above the threshold $(p,q)=(3,6)$. The negative (unstable) eigenvalue exist for every $\gamma \in (0,\gamma_{3,4})$ for $(p,q) = (3,4)$ but does not exist near $\gamma_{3,6}$ for $(p,q) = (3,6)$. 

\item If $(p,q)=(3,5)$ codes with equal numbers of symbols $A$ and $a$ (any combinations and any signs are admissible) have a multiple zero eigenvalue by Lemma \ref{lem-3-5}. The stability analysis is inconclusive, since the multiple zero eigenvalue can split for any $\varepsilon>0$. These codes are highlighted in red in Table \ref{tab:3-5}.
\end{itemize}

Figure~\ref{Fig:Stacked_3-4} shows eigenvalues of the truncated spectral problem (\ref{Gener-matrix}) versus $\gamma$ in $(0,\gamma_{3,4})$ for $(p,q)=(3,4)$. 
The three panels correspond to the codes $\mathcal{A}_{5,4}^-$, $\mathcal{A}_{5,4}^+$, and $\mathcal{A}_{5,5}^-$. 

\begin{itemize}
	\item The code $\mathcal{A}_{5,4}^-$ is stable for $\gamma$ near $\gamma_{3,4}$ but becomes unstable for intermediate values of $\gamma$ due to coalescence of real eigenvalues forming complex pairs. As $\gamma$ becomes small, all splitting have been resolved and the code becomes stable again. 
	
	\item The code $\mathcal{A}_{5,4}^+$ is unstable for any $\gamma \in (0,\gamma_{3,4})$ due to the negative eigenvalue. Nevertheless, there exist additional complex eigenvalues for the intermediate values of $\gamma$ including the values near $\gamma_{3,4}$. 
	
	\item The code $\mathcal{A}_{5,5}^-$ is unstable near $\gamma_{3,4}$ due to complex eigenvalues, which persist towards smaller values of $\gamma$ but reappear back as positive  eigenvalues for $\gamma$ near $0$. 
\end{itemize}

\begin{figure}%[htb!]
	\centerline{\includegraphics[width=0.95\textwidth,height=0.85\textheight]{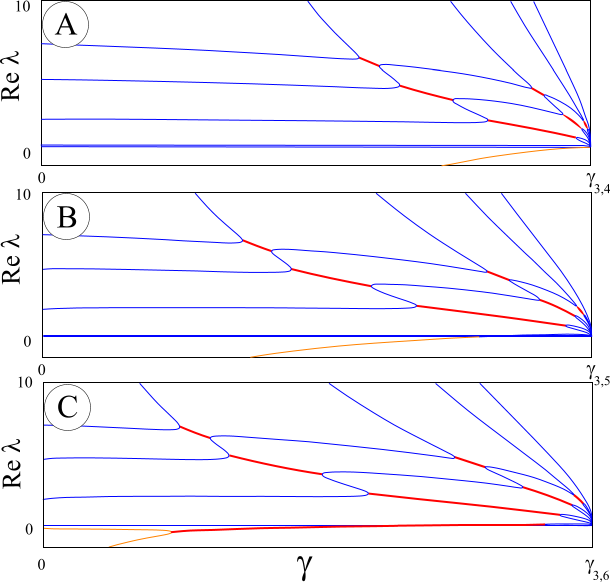}}
	\caption{Eigenvalues of the truncated spectral problem (\ref{Gener-matrix}) versus $\gamma$ in $(0,\gamma_{p,q})$ for the code $\mathcal{A}_{5,5}^+$ 
	with $(p,q)=(3,4)$ (panel A), $(p,q)=(3,5)$ (panel B), and $(p,q)=(3,6)$ (panel C). The real parts of the eigenvalues are shown: orange for real negative eigenvalues, blue for real positive eigenvalues, and red (bold) for complex eigenvalues.}
	\label{Fig:5_5+}
\end{figure}

\begin{figure}%[htb!]
	\centerline{\includegraphics[width=0.95\textwidth,height=0.85\textheight]{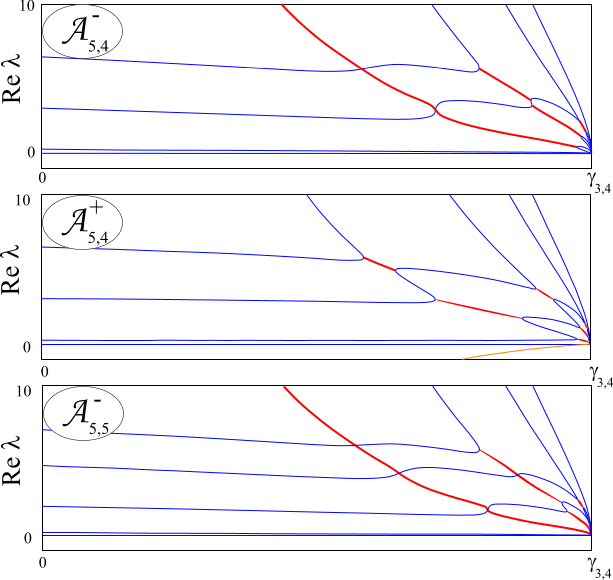}}
	\caption{Eigenvalues of the truncated spectral problem (\ref{Gener-matrix}) versus $\gamma$ in $(0,\gamma_{3,4})$ with $(p,q) = (3,4)$ for the codes $\mathcal{A}_{5,4}^-$ (upper panel), $\mathcal{A}_{5,4}^+$ (middle panel),  and $\mathcal{A}_{5,5}^-$ (lower panel). The color scheme is the same as in Fig.~\ref{Fig:5_5+}.}
	\label{Fig:Stacked_3-4}
\end{figure}

\section{Conclusion}\label{Sect:Conclusion}

We have analyzed spectral stability of intrinsic localized modes (ILMs) in the DNLS equation with competing power nonlinearities. The analysis holds in the anticontinuum limit (ACL) and relies on the count of eigenvalues of the truncated generalized eigenvalue problem and their persistence as eigenvalues 
of the spectral stability problem. In addition, we also analyzed eigenvalues 
of the Hessian operators associated with the variational characterization 
of ILMs and computed minimizers and constrained minimizers of energy 
in the ACL.

The outcome of this work shows a complicated pattern of stability of ILMs 
depending on the strength parameter $\gamma$ of the competing nonlinearities 
with powers $(p,q)$, which is defined in $(0,\gamma_{p,q})$. We identified 
the universally stable codes $\mathcal{A}_A$ and $\mathcal{A}_a$ in (\ref{Eq:uniform}), where $\mathcal{A}_A$ corresponds to the local energy minimizers for any length $N$ and $\mathcal{A}_a$ corresponds to a local constrained energy minimizer for $N = 1$. In addition, we studied stability 
of the codes for stacked modes $\mathcal{A}_{n,m}^+$ and $\mathcal{A}_{n,m}^-$ in (\ref{Eq:Stacked}) and found universal stability of $\mathcal{A}_{k+1,k}^-$ for $N = 2k+1$ for the values of $\gamma$ near $\gamma_{p,q}$. Additionally, the codes $\mathcal{A}_{k,k}^+$ for $N = 2k$ are also stable for the values of $\gamma$ near $\gamma_{p,q}$ but this stability holds for $q > 7$ if $p = 2$, for $q > 5$ if $p = 3$, and for $q > p$ if $p \geq 4$. The asymptotic analysis 
of the spectral stability problem in the limit $\gamma \to \gamma_{p,q}$ is an open question for further studies.

We also observed that eigenvalues of the spectral stability problem 
are very different in magnitude in the limit $\gamma \to 0$ and this 
explains the appearance of many stable codes for small values of $\gamma$ 
especially for physically relevant cases $(p,q) = (2,3)$ and $(p,q) = (3,4)$. 
The asymptotic analysis in the limit $\gamma \to 0$ is another open question 
for further studies. 

Finally, variational characterization of global minimizers of energy 
and constrained minimizers of energy is also an interesting mathematical 
problem for further studies, both in the ACL and for other values of $\varepsilon > 0$.

\end{document}